\pdfoutput=1

\def\editmode{0}
\def\reportmode{0}

\def\bibfilenames{WISENET}

\if\reportmode1
\documentclass[10pt,final,onecolumn]{IEEEtran}
\else
\documentclass[journal]{IEEEtran}
\fi
\usepackage{dsfont}
\if\editmode1  
\usepackage{mathrsfs}
\newcommand{\acom}[1]{\noindent\textcolor{red}{{[#1]}}} 
\usepackage{bbm}
\usepackage{amssymb}
\usepackage[english]{babel}
\usepackage[utf8]{inputenc}
\usepackage{algorithm}
\usepackage[noend]{algpseudocode}
\floatname{algorithm}{Procedure}
\usepackage[backend=bibtex,style=alphabetic,sorting=debug]{biblatex}
\DeclareFieldFormat{labelalpha}{\thefield{entrykey}}
\DeclareFieldFormat{extraalpha}{}
\bibliography{\bibfilenames}
\newcommand{\cmt}[1]{\noindent\textcolor{lightgreen}{\underline{[#1]}}} 
\newcommand{\hc}[1]{\textcolor{blue}{#1}} 
\newenvironment{myitemize}{\begin{itemize}}{\end{itemize}}
\newcommand{\myitem}{\item}

\newtheorem{theorem}{Theorem}
\newcommand{\timevarhindsight}{\tbm a_n^\circ}
\newcommand{\energybound}{B_y}
\newcommand{\rev}{}

\newcommand{\nReg}{\rev {\Omega}^{(n)}}
\newcommand{\strongcvxparf}{\beta_{\ell}}

\else
\newcommand{\acom}[1]{\noindent\textcolor{red}{{[#1]}}} 
\newcommand{\nextver}[1]{\noindent\textcolor{orange}{{[#1]}}} 
\renewcommand{\acom}[1]{\ignorespaces}
\renewcommand{\nextver}[1]{\ignorespaces}

\usepackage{mathrsfs}
\usepackage{amssymb}
\usepackage[english]{babel}
\usepackage[utf8]{inputenc}
\usepackage{algorithm}
\usepackage[noend]{algpseudocode}
\floatname{algorithm}{Procedure}
\usepackage{cite}
\newcommand{\cmt}[1]{} 
\newcommand{\hc}[1]{\textcolor{black}{#1}} 

\newcommand{\energybound}{\hc{B}_y}

\newcommand{\nReg}{\rev {\Omega}^{(n)}}
\newcommand{\strongcvxparf}{\beta_{\ell}}

\newcommand{\timevarhindsight}{\tbm a_n^\circ}
\newenvironment{myitemize}{}{}
\newcommand{\myitem}{}

\newtheorem{theorem}{Theorem}
\newcommand{\rev}{}

\fi 

\usepackage{graphicx}

\usepackage{subcaption}              

\usepackage[utf8]{inputenc}

\usepackage{amsfonts}
\usepackage{amsmath}
\usepackage{mathtools}

\usepackage{amssymb}

\usepackage{bm}

\usepackage{color,verbatim}
\usepackage{multirow}
\usepackage{accents}

\usepackage{theoremref}

\usepackage{url}

\newcounter{rulecounter}
\newcommand{\resetrule}{ \setcounter{rulecounter}{0}}
\resetrule

\newsavebox{\selvestebox}
\newenvironment{colbox}[1]
  {\newcommand\colboxcolor{#1}%
   \begin{lrbox}{\selvestebox}%
   \begin{minipage}{\dimexpr\columnwidth-2\fboxsep\relax}}
  {\end{minipage}\end{lrbox}%
   \begin{center}
   \colorbox{\colboxcolor}{\usebox{\selvestebox}}
   \end{center}}

\definecolor{orange}{rgb}{1,0.8,0}
\definecolor{gray}{rgb}{.9,0.9,0.9}
\definecolor{darkgray}{rgb}{.3,0.3,0.3}
\definecolor{darkblue}{rgb}{.1,0.0,0.3}
\definecolor{lightblue}{rgb}{0.7,0.7,1}
\definecolor{lightred}{rgb}{1,0.7,.7}
\definecolor{purple}{RGB}{204,153,255}
\definecolor{lightgray}{rgb}{.95,0.95,0.95}
\definecolor{lightgreen}{rgb}{0.3,0.5,0.3}
\definecolor{darkgreen}{rgb}{0.05,0.3,0.05}

\newcommand{\tbm}[1]{{\tilde{\bm #1}}}

\newcommand{\hbm}[1]{{\hat{\bm #1}}}

 \newcommand{\define}{\triangleq}

\newtheorem{myproposition}{Proposition}
\newtheorem{myremark}{Remark}
\newtheorem{myproblemstatement}{Problem Statement}
\newtheorem{mylemma}{Lemma}
\newtheorem{mytheorem}{Theorem}
\newtheorem{mydefinition}{Definition}
\newtheorem{mycorollary}{Corollary}

\usepackage{enumerate}

\begin{document}

	\title{Online Joint Topology Identification and Signal Estimation 
	from Streams with Missing Data }
	\author{
    \IEEEauthorblockN{Bakht Zaman\IEEEauthorrefmark{2},~\IEEEmembership{Member,~IEEE}, Luis Miguel Lopez-Ramos\IEEEauthorrefmark{1}\IEEEauthorrefmark{2},~\IEEEmembership{Member,~IEEE,}\\ and Baltasar Beferull-Lozano\IEEEauthorrefmark{1}\IEEEauthorrefmark{2},~\IEEEmembership{Senior Member,~IEEE}\\
    \IEEEauthorblockA{\IEEEauthorrefmark{1}WISENET Center, Department of ICT, University of Agder, Grimstad, Norway}\\
    \IEEEauthorblockA{\IEEEauthorrefmark{2}	Simula Research Laboratory, Simula Metropolitan Center for Digital Engineering, Oslo, Norway}
}
	\thanks{The work in this paper was supported by the SFI Offshore Mechatronics grant 237896/E30, the PETROMAKS Smart-Rig grant 244205, the IKTPLUSS INDURB grant 270730/O70, and the IKTPLUSS DISCO grant 338740 from the Research Council of Norway.}
	\thanks{B. Beferull-Lozano is with the WISENET Center, Dept. of ICT, University of Agder, Jon Lilletunsvei 3, Grimstad, 4879 Norway. e-mail:baltasar.beferull@uia.no; and also with the SIGIPRO Department, Simula Metropolitan Center for Digital Engineering, e-mail: baltasar@simula.no. B. Zaman and L. M. Lopez-Ramos were with the WISENET Center at the time this work was completed. Now B. Zaman is with Simula Research Laboratory. E-mail: bakht@simula.no; and L. M. Lopez-Ramos is with the Holistic Systems Department, Simula Metropolitan Center for Digital Engineering, e-mail: luis@simula.no.
 
    This paper has supplementary downloadable material available at http://ieeexplore.ieee.org., provided by the author. The material includes several proofs and one figure. Contact bakht@simula.no for further questions about this work.
 }
 \vspace{-5mm}
}
	
	\maketitle
	
	\begin{abstract}
Identifying the topology underlying a set of time series is useful for tasks such as prediction, denoising, and data completion. Vector autoregressive (VAR) model-based topologies capture dependencies among time series and are often inferred from observed spatio-temporal data. When data are affected by noise and/or missing samples, topology identification and signal recovery (reconstruction) tasks must be performed jointly. Additional challenges arise when i) the underlying topology is time-varying, ii) data become available sequentially, and iii) no delay is tolerated. This study proposes an online algorithm to overcome these challenges in estimating VAR model-based topologies, having constant complexity per iteration, which makes it interesting for big-data scenarios. The inexact proximal online gradient descent framework is used to derive a performance guarantee for the proposed algorithm, in the form of a dynamic regret bound. Numerical tests are also presented, showing the ability of the proposed algorithm to track time-varying topologies with missing data in an online fashion.

\end{abstract}
\section{Introduction}\label{s:intro}
\cmt{Motivation Online Topology Id}%
In many applications involving complex systems, causal relations among time series are computed and encoded as a graph, where each node corresponds to a time series, and often reveals the topology of an underlying social, biological, or brain network~\cite{kolaczyck2009}. A causality graph provides insights into the complex system under analysis and enables certain tasks, such as forecasting~\cite{isufi2018forecasting}, signal
	reconstruction~\cite{lorenzo2016lms}, anomaly
	detection~\cite{liu2016unsupervised}, and
	dimensionality
	reduction~\cite{shen2017dimensionalityreduction}. 
	The assumption that the interaction patterns among variables remain unchanged does not always hold, and future data may have different underlying properties than the current and historic data. This situation, known as \emph{concept drift}~\cite{hoeltgebaum2021estimation}, invalidates methods that assume stationarity. Moreover, when data are not all available at once but sequentially, batch processing is not possible, calling for \emph{online} algorithms that continuously update model parameters after receiving each data sample.
	
	\cmt{Motivation of the missing values and noisy observations}It is impractical to assume that the data are fully observable at every node and time instant~\cite{little2014,pavez2019missingdata}, for diverse reasons. Data acquired by a sensor network may be partially observed due to faulty sensors, network congestion, or sporadic observation due to energy constraints. Data may only be partially available due to variable environmental factors in ecological networks 
 \cite{humbert2009better,clark2004population} or due to privacy reasons in social networks. 
 Missing values in econometrics time series are considered in \cite{harvey1984estimating}, where some of the monthly-sampled variables are missing for some months in an earlier period.
General types of missing value patterns in the data and their estimations are presented in detail in \cite{little2019statistical}.
In spatio-temporal modeling, \cite{grover2021MissingDataSpatioTemporal} deals
with the problem of missing data and proposes an online
algorithm for estimating missing values. 
Practical applications of the estimation of missing data include \cite{adhikari2022missingdataIoT} in the context of IoT, 
\cite{kim2022missingvaluesship} 
for ship operational data in the maritime transportation domain, and \cite{zhang2022missingenvironmental} in the application of a real-time water quality monitoring system. Similarly, in an industrial environment, \cite{pan2022imputation} presents a deep learning-based algorithm to address missing data. \textcolor{black}{Finally, in power systems, the topology, as well as the state, are jointly estimated due to missing data in \cite{karimi2021joint}.}

This paper addresses the problem of estimating topologies from time series in an online fashion where the data contain missing values.
Regarding related work, we initially present methods for topology identification under complete data, and later we discuss those works under noisy and missing data.
Identifying graphs capturing spatio-temporal
``interactions'' among time series has attracted significant 
attention in the literature \cite{mateos2018connecting}.
\begin{myitemize}%
\myitem\cmt{Memoryless interactions}%
\begin{myitemize}%
	\myitem\cmt{undirected graphs}%
	\begin{myitemize}%
		\myitem\cmt{refs}%
		\begin{myitemize}%
			\myitem\cmt{correlation}For undirected topologies, correlation and%
			\myitem\cmt{partial correlations} partial correlation graphs~\cite{kolaczyck2009},  Markov random fields~\cite{angelosante2011graphical}, and graph signal based approaches~\cite{segarra2017templates} are used.
		\end{myitemize}%
		\myitem\cmt{limitations}For directed graphs, 
	\end{myitemize}%
	\myitem\cmt{directed graphs}%
	\begin{myitemize}%
		\myitem\cmt{structural equation models (SEMs)}%
		structural equation models (SEM)~\cite{kline2015},  \cite{shen2017tensor}
		\myitem\cmt{Bayesian networks}or Bayesian networks~\cite[Sec.~8.1]{bishop2006} are mainly employed. 
	\end{myitemize}%
	\myitem\cmt{limitations}However, these methods only account
	 for \emph{memoryless} interactions, i.e., they cannot accommodate delayed causal interactions,
	where the value of a time series at a given time instant is
	related to the past values of other time series.
\end{myitemize}%
\par 
\myitem\cmt{Memory based approaches}%
\begin{myitemize}%
	\myitem\cmt{Granger causality notion}%
	\begin{myitemize}%
		\myitem\cmt{def}A notion of causality among time series is due to Granger~\cite{granger1988causality} based on the optimal prediction error,
		which is generally difficult
		to determine optimally \cite[p. 33]{zellner1979causality}, \cite{kay1}.
	\end{myitemize}%
	\myitem\cmt{VAR}%
	\begin{myitemize}%
		\myitem\cmt{VAR causality}Thus, alternative causality definitions
		based on, e.g., vector autoregressive (VAR) models are typically preferred~\cite{goebel2003varcausality,basu2015granger}.
		\myitem\cmt{Gaussian and stationary}VAR topologies are estimated
		assuming stationarity and a Gaussian distribution of the innovations  in~\cite{bach2004learning,songsiri2010selection}, and additionally
		\myitem\cmt{sparsity}%
		\begin{myitemize}%
			\myitem \cmt{Group-Lasso Estimator}assuming sparsity in~\rev{\cite{bolstad2011groupsparse,songsiri2013vargranger,mei2017causal}}.
		\end{myitemize}%
	\end{myitemize}%
\cmt{Literature on time-varying topology ID}%
\begin{myitemize}%
\myitem\cmt{motivation}All these approaches assume a model that does not
change over time. Time-varying topologies for undirected graphs include the approaches in ~\cite{kolar2010estimating,yamada2020timevarying} 
\rev{and for directed graphs in \cite{lopezramos2018dynamic}}.
Moreover, an algorithm to jointly estimate multiple graphs representing complex topological patterns is detailed in \cite{yuan2021joint}.
\end{myitemize}%

\cmt{Literature on Online topology est}%
\begin{myitemize}%
\par 
\myitem\cmt{motivation}
All the previously discussed approaches process the entire dataset at once and cannot deal with streaming data due to computational complexity.    
To tackle these issues, in \emph{online} optimization, an estimate is refined with
every new data instance.
\myitem	\cmt{Refs }%
\begin{myitemize}%
	\myitem\cmt{no memory}%
	\begin{myitemize}%
		\myitem \cmt{graphical + SEM}Existing online topology
		identification algorithms include \cite {hallac2017network,shen2017tensor,baingana2014trackingcascades,zaman2020dynamic,shafipour2019onlinetopology, zhang2021online} for memoryless interactions, topology identification in matrix-valued time series \cite{jiang2021online},
		\myitem \cmt{memory-based}and \cite{shen2018online} for nonlinear memory-based dependencies.
	\end{myitemize}%
	
\end{myitemize}%
\end{myitemize}%

Topology identification becomes challenging for noisy data. In \cite{liu2019smoothgraphlearning}, joint signal estimation and topology identification are pursued based on a spatio-temporal smoothness-based graph learning algorithm. The problem becomes even more challenging when data are incomplete. Several batch approaches to identify topologies in the presence of noisy data with missing values are available for undirected topologies in~\cite{berger2020efficient} and for VAR-based directed topologies in~\cite{rao2017estimation,loh2012missingdatavar}. In addition, \cite{coutino2021state} explores topology identification under partial observability of an input signal when an interaction model and its evolution over time are considered. For graph signals with missing values, distributed algorithms are presented in \cite{jiang2020recovery} to recover the signal from noisy observations without topology estimation.
 \myitem \cmt{online approaches}Online prediction of time series with missing data is considered in \cite{anava2015online} 
and \cite{yang2019online}, where the missing values are imputed by their estimates. Theoretical guarantees are presented; however, these works adopt a univariate autoregressive (AR) process model and thus do not extract information about the relations among multiple time series. Moreover, these works consider a static (stationary) model and analyze static regret.\footnote{The regret is an objective performance metric that allows comparing online algorithms (see Sec \ref{sec:composite} for details).} Joint estimation of signal and topology is considered in \cite{ioannidis2019semiblindinference} for a structural VAR model (SVARM) when the observations contain noisy and missing values.
However, no performance guarantees showing the tracking capabilities of the proposed online algorithm are presented.
\end{myitemize}%
\end{myitemize}%
In \cite{zaman2019online}, an online algorithm for topology identification is proposed, where a recursive least squares (RLS)-based loss function helps to improve its tracking capabilities and enables derivation of a sub-linear regret bound. However, the latter algorithm cannot be applied directly when the data are corrupted by noise and missing values, \textcolor{black}{and the present paper proposes a methodology that allows to apply an RLS-based loss function under the aforementioned circumstances. A list of the main differences between the present paper and \cite{zaman2019online} follows:
\begin{itemize}
    \item The input data in this paper contain missing values and are corrupted by noise, implying that the input to the algorithm in \cite{zaman2019online} differs from that presented in this paper. 
    \item The resulting problem formulation in this paper is different, which can be confirmed by comparing (3) in \cite{zaman2019online} and \eqref{eq:jointbatchproblem} here. One major difference is in the sets of optimization variables, as \eqref{eq:jointbatchproblem} in this paper optimizes also over the reconstructed signal values, in addition to the VAR parameters. 
    \item 
    Decomposability across nodes [cf. (3) in \cite{zaman2019online}] results in a separate optimization problem for each node. However, the problem formulated here is not separable because of the coupling introduced by the optimization variables intended to estimate the missing signal values.
    \item Missing data makes the analysis different from the one in \cite{zaman2019online}. The bound in Theorem 1 here depends on the bound on the gradient derived in Lemma 1, which does not appear in \cite{zaman2019online}. Moreover, the bound in Corollary 1 also depends on the bound of the error associated with inexact gradients, derived in Lemma 3.
\end{itemize}
}

This paper proposes an online algorithm to estimate time-varying, memory-aware causality graphs from streaming time series that are affected by noise and missing data while reconstructing the input signals by denoising and imputation of missing values.  
The contributions are:
C1. The formulation of the online estimation and reconstruction task as a sequential decision problem, to account for the impact of decisions in the future stages. More specifically, a sequential cost function inspired from \cite{ioannidis2019semiblindinference} is put forth, involving signal mismatch from both the noisy samples and the current prediction, time-variation of the estimated topology parameter estimates, and a sparsity-promoting term. 
C2. The application of well-justified simplifying assumptions to the cost defined in C1 introduces a loss function that can be tackled using an online convex optimization approach. Based on this, an online algorithm is proposed, named \emph{joint signal and topology identification via recursive sparse online learning} (JSTIRSO), which has 
tracking capability. The loss function that JSTIRSO optimizes is augmented with an additional term inspired by recursive least squares (RLS), which not only helps in tracking capability but also enables theoretical analysis. The proposed algorithm has fixed computational complexity per sample, which is suitable for big data applications. 
C3. The derivation of a dynamic regret bound, to characterize the performance of JSTIRSO when the topology is time-varying. The derived dynamic regret bound depends on the properties of the data, the error due to missing values, and the parameters of the algorithm. Moreover, the error of JSTIRSO in time-varying scenarios is quantified in terms of the data properties.
C4. Finally, the empirical validation of the performance of the proposed algorithms through numerical tests.
\par
The rest of the paper is organized as
follows: Sec. \ref{sec:model} presents the model and a batch formulation for tracking of VAR causality graphs. Sec. \ref{sec:online} introduces the sequential joint tracking and signal estimation and reviews the online convex optimization approach. To solve the sequential problem in an online fashion, an approximate loss function is obtained in Sec. \ref{sec:approximateLoss}, and an intermediate algorithm is derived. An alternative loss function is presented and used to derive the JSTIRSO in Sec. \ref{sec:tirso}, 
which in turn is characterized analytically (dynamic regret analysis) in Sec. \ref{sec:analysis}. Numerical results are presented in Sec. \ref{sec:simulationsE}, and Sec. \ref{sec:conclusions} concludes the paper. \\
\cmt{Notation}\textbf{Notation.} Bold lowercase (uppercase) letters
denote column vectors (matrices). Operators 
$\mathbb E[\cdot]$,
\rev{$\partial$},
$(\cdot)^\top $, 
$\mathrm{vec}(\cdot)$, 
and
$\lambda_{\mathrm{max}}(\cdot)$, 
respectively denote 
expectation, 
\rev{sub-differential}, 
matrix transpose, 
vectorization, 
and
the maximum eigenvalue of a matrix. 
The operator $\nabla$ denotes a gradient and $\nabla^s$ represents a subgradient. Symbols $\bm 0_N$ and $\bm 0_{N\times N}$, represent all-zero vector and matrix, $\bm 1_N$ all-ones vector,  and $\bm I_{N}$ identity matrix, all of the given size. 
Finally, $[\cdot]_+ \triangleq \mathrm {max} (\cdot, 0)$, and $\mathds{1}$ is the indicator satisfying $\mathds 1 \{x\}=1$ if $x$ is true and
$\mathds 1 \{x\}=0$ otherwise.
\section{Model and Problem Formulation}
\label{sec:model}
\label{sec:problem_formulation}
\begin{myitemize}
\myitem \cmt{Time series and causality notion:}	
\begin{myitemize}%
	\myitem \cmt{Time series}Consider a collection of $N$ time
	series, where $y_n[t]$, $t=0,1, \ldots , T-1$, denotes the
	value of the $n$-th time series at
	time $t$. 
	\myitem\cmt{Goal: directed causal
		graphs}A causality graph   $\mathcal
	G\triangleq(\mathcal V, \mathcal E)$ is a graph 
	where
	\begin{myitemize}%
		\myitem\cmt{vertices}the $n$-th vertex in $\mathcal V=\{
		1,\ldots,N\}$  is identified with the $n$-th time
		series $y_n[t]$ and
		\myitem\cmt{edges}there is an edge (or
		arc) from $n'$ to $n$ $((n,n')\in  \mathcal E)$ if and only if (iff)
		$y_{n'}[t]$ \emph{causes} $y_{n}[t]$ according to a certain
		causality notion.            \end{myitemize}%
\end{myitemize}%
\end{myitemize}%
\begin{myitemize}%
\myitem \cmt{Dynamic model: Time-varying VAR process}%
A prominent notion of
causality can be defined using VAR models. \textcolor{black}{Moreover, if the topology is dynamic, a time-varying VAR model can be defined. }
\begin{myitemize}%
	\myitem\cmt{def}To this end,
	consider the order-$P$ time-varying 
	VAR model \cite{lutkepohl2005}:
	\vspace{-2mm}
 \textcolor{black}{
	\begin{equation}\label{eq:model}
	\bm y[t] =\sum_{p=1}^{P}\mathbf A_p^{(t)}
	 \bm y[t-p] +\bm u[t],
	\end{equation}
 }
	where 
	\begin{myitemize}%
		\myitem$\bm y[t]\triangleq [y_1[t], \ldots, y_N[t] ]^\top$,
		\myitem$\bm A_p^{(t)}
		 \in \mathbb R
		^{N \times N}, p=1, \ldots, P$, are the matrices of
		time-varying 
		VAR parameters and
		\myitem$\bm u[t]\triangleq [u_1[t],\ldots,u_N[t]]^\top$ is the \emph{innovation process}, generally assumed
		to be a temporally white, zero-mean stochastic process, i.e.,
		$\mathbb E [\bm u[t]  ]=\bm 0_N$ and $\mathbb E [\bm
		u[t]\bm u^\top[\tau]  ]=\bm 0_{N\times N}$ for $t\ne \tau$.
        The parameters $\{\bm A_p^{(t)}\}_{p=1}^P$ follow a certain law of motion such as introduced in \cite[Ch. 18]{kilian2017}.
		
	\end{myitemize}%
	\myitem\cmt{LTI interpr}With $a_{n,n'}^{(p)(t)}$ the  $n,n'$-th entry of $\bm A_p^{(t)},
	$ 
	\eqref{eq:model} becomes
	\begin{align} \nonumber
	y_n[t]&=\sum_{n'=1}^{N}\sum_{p=1}^Pa_{n,n'}^{(p)(t)}y_{n'}[t-p]+u_n[t]
	\\ \label{eq:fnE}
	&= \sum_{n'\in \mathcal{N}(n)}\sum_{p=1}^P a_{n,n'}^{(p)(t)}y_{n'}[t-p]+ u_n[t],
	\end{align}
	for $n=1, \ldots, N$, 
	where
	$\mathcal{N}(n, t)~ \triangleq \{ n': \bm a_{n,n'}^{(t)}\neq \bm 0_P \}$ \nextver{is the \emph{in-neighborhood} of node $n$,} and 
	\begin{equation}\label{eq:def-an}
	\bm a_{n,n'}^{(t)}\define \left [a_{n,n'}^{(1)(t)},\ldots,a_{n,n'}^{(P)(t)}\right]^\top.
	\end{equation}

		\cmt{VAR causality} This model introduces the concept of \emph{VAR causality} \cite{geiger2015causalinferenceVAR}, with a similar spirit as of Granger causality, but less challenging to compute. Given a process order $P$, the time series $y_i[t]$ \emph{VAR-causes} time series $y_j[t]$ iff the $P$ most recent values of $y_i[t]$ carry information that reduces the prediction mean square error (MSE) of $y_j[t]$, see \cite{zaman2019online} for a detailed discussion. 	
	
\end{myitemize}

\myitem\cmt{VAR causality}

\begin{myitemize}%
	\myitem\cmt{intuition}When
	$\bm u[t]$ is a zero-mean and temporally
	white stochastic process,  the term
	$\hat
	y_n[t]\define \sum_{n'\in \mathcal{N}(n)}\sum_{p=1}^P
	a_{n,n'}^{(p)(t)}y_{n'}[t-p]$
	in \eqref{eq:fnE} is the \emph{minimum
		mean square error estimator} of $y_n[t]$
	given the previous values of all time
	series
	$\{y_{n'}[\tau],~n'=1,\ldots,N,~\tau <
	t\}$; see e.g.~\cite[Sec. 12.7]{kay1}. \textcolor{black}{ The set
	$\mathcal{N}(n,t)$ therefore collects
	the indices of those time series that
	participate in this optimal
	predictor of $y_n[t]$; in other words, the information
	provided by time series $y_{n'}[t]$
	with $n'\notin\mathcal{N}(n,t)$ is not
	informative to predict
	$y_n[t]$.  \myitem\cmt{def}This allows us to express the definition of 
	VAR causality in a clearer and more compact way:
	$y_{n'}[t]$ \emph{VAR-causes}
	$y_{n}[t]$ {around time instant $t$} whenever
	$n'\in \mathcal{N}(n,t)$. Equivalently,
	$y_{n'}[t]$ \emph{VAR-causes} $y_{n}[t]$ {around time instant $t$}
	 if $\bm a_{n,n'}^{(t)}\neq\bm 0_P$.
	\myitem\cmt{relation to Granger}%
	\myitem\cmt{graph}VAR causality relations among
	the $N$ time series can be represented  using a \textcolor{black}{time-varying} 
	causality graph where
	\begin{myitemize}%
		\myitem\cmt{edges}$\mathcal{E}(t)\triangleq \{(n,n'):~\bm
		a_{n,n'}^{(t)}\neq \bm 0_P\}$.\myitem\cmt{neighborhood} Clearly, in
		such a graph, $\mathcal{N}(n,t)$ is the
		in-neighborhood of node $n$. 
		\myitem\cmt{weights}To
		quantify the strength of the
		causality relations, a weighted graph
		can be constructed by assigning, e.g., the weight $\|\bm a_{n,n'}^{(t)}\|_2$
		to the edge~$(n,n')$.%
	\end{myitemize}%
 }
\end{myitemize} 
\par

\myitem\cmt{Batch problem and its nature}With these definitions,  
\end{myitemize}%
\cmt{Batch Problem without missing values}
the inference problem can be formally stated as:
\begin{myitemize}%
\myitem \cmt{Given}given the observations $\{\bm y[t]\}_{t=0}^{T-1}$ (in \emph{batch} form) and the VAR process order, $P$,
\myitem \cmt{Requested}find the time-varying VAR coefficients  $\{\{ \bm A_p^{(t)}\}_{p=1}^P\}_{t=P}^{T-1}$ such that it yields sparse topology at each time instant.
\myitem \cmt{Assumption on variations}Without assumptions on the variations of the topologies, the problem involves more unknown variables than the available data samples and is ill-posed. In this case, we assume that the variations in the topology are constrained so that the cumulative norm difference between consecutive sets of parameters does not exceed a given budget of $B$. 
\end{myitemize} 
\cmt{Batch Solution: Estimating time-varying coefficients without missing values}The formulation in \cite{bolstad2011groupsparse} can be extended to a time-varying model as follows:
\begin{subequations}
\label {eq:batchproblem}
\begin{align} 
&\underset{\{ \{ \bm A_p^{(\tau)} \}_{p=1}^{P}\}_{\tau=P}^{T-1}}{\arg \min} \frac{1}{\,2(T-P)}  \sum_{t=P}^{T-1}\left \lVert \bm y[t]-\sum_{p=1}^{P} \bm A_p^{(t)} \, \bm y[t-p]\right \rVert_2^2 \nonumber\\
&\quad \quad \quad \quad \quad \quad \quad +\sum_{t=P}^{T-1}\Omega\left(\left \{\bm A_p^{(t)}\right \}_{p=1}^{P}\right) \label{eq:batch_objective} 
\\
&\text{s. t.} \sum_{t=P+1}^{T-1} \left \lVert \mathrm {vec}\left (\{ \bm A_p^{(t)} \}_{p=1}^{P} \right ) - \mathrm {vec}\left (\{ \bm A_p^{(t-1)} \}_{p=1}^{P} \right ) \right \rVert_2^2 \leq B, \label{eq:batch_constraint}
\end{align}
\end{subequations}
where
\begin{myitemize}%
\myitem \cmt{Least squares loss function}the first term in the cost function is the least-squares loss, and 
\myitem \cmt{Sparsity promoting regularization function}the second term is a group sparsity-promoting regularization function defined as
\begin{equation}
\label{eq:omega_definition}
	\Omega\left(\left \{\bm A_p^{(t)}\right \}_{p=1}^{P}\right) 
	\define 
	\lambda \sum_{n=1}^{N}\sum_{n'=1}^{N} \mathds{1} \{ n' \neq n\}\big \lVert \bm a_{n,n'}^{(t)} \big \rVert_2,
\end{equation}
where $\bm a_{n,n'}^{(t)}$ has the same structure as \eqref{eq:def-an} with time-varying VAR parameters.
 The regularization function $\Omega$ promotes sparse edges in the causality graphs. 
\myitem \cmt{Parameter $\lambda$}The parameter $\lambda$ is a user-defined constant that controls the sparsity in the edges of the graph. 
\myitem \cmt{Budget on the path length}The constraint \eqref{eq:batch_constraint} 
restricts the amount of variation in the VAR parameters, and is necessary for the problem to have a meaningful solution (otherwise it would be ill-posed\nextver{avoid ill-conditioning}). 
\end{myitemize} 
\par 
\cmt{online problem statement with missing values}To formulate the problem of estimating the causality graphs when observations are affected by noise and some values are missing, consider a subset of $\mathcal V$ where the signal is observed, given by $\mathcal M_t \subseteq \mathcal V$.
\cmt{Modelling missing values with Bernoulli random variable}The (random) pattern of missing values is collected in the masking vector $\bm m[t] \in \mathbb{R}^N$
where $m_n[t], n=1,\ldots, N$, are i.i.d. Bernoulli random variables taking value 1 with probability $\rho$ and 0 with probability $1-\rho$. 
Let $\tbm y[t]$ be the observation obtained at time $t$, given by: 
\vspace{-2mm}
\begin{equation} \label{eq:observationmodel}
\tbm y[t]= \bm m[t] \odot ( \bm y[t] + \bm \epsilon[t]), 
\end{equation} 
where $\odot$ denotes element-wise product, and $\bm \epsilon[t]$ is the observation noise vector. 

\cmt{Batch problem with missing values}In a batch setting, the problem of estimating time-varying topologies with missing values is:
\begin{myitemize}
\myitem \cmt{Given}given the noisy observations $\{\tbm y[t]\}_{t=0}^{T-1}$ with missing values, and the VAR process order $P$,
\myitem \cmt{Requested}find the coefficients $\{\{ \hat {\mathbf{\bm A}}_p^{(t)} \}_{p=1}^{P}\}_{t=P}^{T-1}$ such that it yields a sparse topology.
\end{myitemize}%
\cmt{Estimating the signal is necessary for estimating the topology} Since the time series follow a VAR model, the topology can be estimated directly from the observation vector if the missing values are reconstructed (imputed), and the VAR parameters help in such reconstruction. 

\cmt{Jointly batch estimation problem}Thus, a natural approach is to jointly estimate the signals and the VAR coefficients. To this end, the approach advocated in \cite{ioannidis2019semiblindinference} is to solve the following problem: 
\begin{multline} \label {eq:jointbatchproblem}
\left \{   \hbm y[t], \left \{\mathbf{\hat A}_p^{(t)} \right \}_{p=1}^{P}\right \}_{t=P}^{T-1}   =\\
 \underset{\left\{\bm y[t], \{ \bm A_p^{(t)} \}_{p=1}^{P}\right\}_{t=P}^{T-1}}{\arg \min} \frac{1}{2}  \sum_{t=P}^{T-1}\left \lVert \bm y[t]-\sum_{p=1}^{P} \bm A_p^{(t)} \, \bm y[t-p]\right \rVert_2^2 \\
 + \sum_{t=P}^{T-1}\frac{\nu}{2 |\mathcal M_t|} \left \lVert \tbm y[t]-\bm m[t] \odot \bm y[t]\right \rVert_2^2 
+\sum_{t=P}^{T-1}\Omega\left(\left \{\bm A_p^{(t)}\right \}_{p=1}^{P}\right) \\
+ \beta \sum_{t=P}^{T-1} \sum_{p=1}^{P}\lVert \bm A_p^{(t)} -  \bm A_p^{(t-1)} \rVert_{\text F}^2,
\end{multline}
where
\begin{myitemize}%
\myitem the first term is a least-squares (LS) fitting error for all time instants (where the $t$-th term in the summation fits the signal based on the $P$ previous observations and the VAR coefficients at time $t$),
\myitem the second term penalizes the mismatch between the observation vector and the reconstructed signal (recall that $|\mathcal M_t|$ is the number of nodes where the signal is observed\footnote{For those time instants where $|\mathcal M_t| = 0$, the term affected by the fraction will not be considered in the optimization, so the division-by-zero error is avoided.}), 
\myitem the third term is a regularization function that promotes sparsity in the edges,
\myitem and the fourth term limits the variations in the coefficients (it comes from the dualization of the constraint in \eqref{eq:batchproblem}).
\myitem The parameter $\nu >0$ is a constant to control the trade-off between the prediction error based on the VAR coefficients and the mismatch between the measured samples and the signal reported after the reconstruction.
\myitem The parameter $\lambda$ controls the sparsity in the edges while $\beta$ controls the magnitude of the cumulative norm of the difference between consecutive coefficients.
\\
\myitem 
\textcolor{black}{\textbf{Remark 1}. The error in the prediction is due to two sources of uncertainty:
observation noise and innovation in the VAR process. The proposed weighted penalty accounts
for both sources of uncertainty, and allows to use both the input data samples and the estimated VAR parameters to provide robustness to noise. Tuning the hyperparameter $\nu$ allows us to find a balance point between trusting the (noisy) data
and matching a VAR process.
}

\myitem The resulting problem in \eqref {eq:jointbatchproblem} is (separately) convex in  $  \{\bm y[t]\}_{t=P}^{T-1} $ and in $\{\{ \bm A_p^{(t)} \}_{p=1}^{P}\}_{t=P}^{T-1}$, but not jointly convex.
\end{myitemize}%
\cmt{Solving the jointly batch estimation problem}A stationary point of \eqref {eq:jointbatchproblem} can be found via alternating minimization \cite[Corollary 1]{ioannidis2019semiblindinference}. Each subproblem in alternating minimization can be solved via proximal gradient descent. 
Next, we describe how to solve this problem in an online fashion for sequential data.
\section{Online Signal Reconstruction and Topology Inference} 
\label{sec:online}

The batch formulation in \eqref{eq:jointbatchproblem} uses information from all time instants to produce a sequence of reconstructed signal values and VAR parameter (topology) estimates. On the other hand, an online formulation should allow us to produce such a sequence with minimum delay and with fixed complexity (at the price of lower accuracy). Specifically, here we are interested in an algorithm that produces an estimate of $\bm y[t]$ and $\{ \bm A_p^{(t)} \}_{p=1}^{P}$ when the partial observation $\tbm y[t]$ is received. 

To this end, we design an online criterion such that its sum over time matches the batch objective in \eqref{eq:jointbatchproblem}. First, define
\newcommand{\dLoss}{\ell} 
\begin{multline}
\label{eq:def_dloss}
	\dLoss_t \left (
			\{ \bm y[\tau]\}_{\tau = t-P}^{t-1}, 
			\bm y[t], 
			\left \{\bm A_p^{(t)} \right \}_{p=1}^{P} 
	\right )
	\define \\
	\frac{1}{2}  \left \lVert \bm y[t]-\sum_{p=1}^{P}  \bm A_p^{(t)} \, \bm y[t-p]\right \rVert_2^2 
	+ \frac{\nu}{2 |\mathcal M_t|} \left \lVert \tbm y[t]-\bm m[t] \odot  \bm y [t]\right \rVert_2^2.
\end{multline}

Now we can use the expression above,\footnote{The splitting of the arguments of $\ell_t$ into the present and past samples will become useful in subsequent sections.} and the definition of $\Omega(\cdot)$ from \eqref{eq:omega_definition}, to define the dynamic cost function:
\begin{multline} \label {eq:dynamic_cost_1}
c_t \left (
	\{ \bm y[\tau]\}_{\tau = t-P}^{t},\left \{\bm A_p^{(t)}\right\}_{p=1}^{P}, \left \{\bm A_p^{(t-1)} \right \}_{p=1}^{P} 
\right ) 
\define \\
\dLoss_t \left (
	\{ \bm y[\tau]\}_{\tau = t-P}^{t-1}, 
			\bm y[t], 
			\left \{\bm A_p^{(t)} \right \}_{p=1}^{P} 
	\right )
+
\Omega\left(\left \{\bm A_p^{(t)}\right \}_{p=1}^{P}\right) 
\\+ \beta \sum_{t=P}^{T-1} \sum_{p=1}^{P}\lVert \bm A_p^{(t)} -  \bm A_p^{(t-1)} \rVert_{\text F}^2.
\end{multline}
The objective function in \eqref{eq:jointbatchproblem} can be rewritten as $\sum_t c_t(\cdot,\cdot,\cdot)$. It becomes clear that producing an estimate of $\bm y[t]$ and $\{ \bm A_p^{(t)} \}_{p=1}^{P}$ does not only have an impact on $c_t(\cdot,\cdot,\cdot)$, but also on $\{c_\tau(\cdot,\cdot,\cdot)\}_{\tau=t}^{t+P}$. Such a coupling in time is taken into account in the framework of dynamic programming (or reinforcement learning), where the goal is to find a policy $\pi$: 
\begin{alignat}{4}
\begin{aligned}
\pi \!& :\mathbb{R}^{PN} \times \mathbb{R}^{N^2 P} \times \mathbb{R}^N  \times \mathbb{R}^{N} 
& \rightarrow&&
& \!\!  \mathbb{R}^N \times \mathbb{R}^{N^2 P} 
\\
\pi \!&\left( \!
	\{ \hbm y[\tau]\}_{\tau = t-P}^{t-1},
	{  \{\hat {\bm A}_p^{(t-1)} \}_{p=1}^{P}}, 
	\tbm y[t],
	\bm m[t] \!
\right) \!\!\! &
 \rightsquigarrow &&
&\!\! \hat {\bm y}[t], \{\hat {\bm{A}}_p^{(t)} \}_{p=1}^{P}
\end{aligned}
\end{alignat}
such that the cumulative cost is minimized in expectation. Learning such a policy (via e.g., deep reinforcement learning) would be computationally intensive and require a high amount of data, and it is left out of the scope of the present paper. Instead, we propose to approximate such a policy using the much more tractable framework of online convex optimization (reviewed next). Fortunately enough, the structure of \eqref{eq:dynamic_cost_1} resembles that of the composite problems that can be efficiently dealt with via proximal online gradient descent (OGD). In the next section, an approximation of the cost function discussed above will be taken in a way such that we can derive a proximal OGD update over $ \{ {\bm A}_p^{(t-1)} \}_{p=1}^{P}$.

In the remainder of this section, the theoretical background of proximal OGD and inexact proximal OGD (\mbox{IP-OGD}) will be introduced. In Sec. \ref{sec:approximateLoss},
we will explain the approximations we take in order to be able to apply the \mbox{IP-OGD}  framework~\cite{dixit2019onlineproximal} to the online problem at hand.
\vspace{-3mm}
\subsection{Theoretical background: composite problems}
\label{sec:composite}
In the sequel, we present a framework to solve composite-objective optimization problems in an online fashion.

\begin{myitemize}
\myitem \cmt{Splitting the Objective} Consider a sequence of functions consisting of a loss and a regularization part. Each function in the sequence is given by:
\begin{equation} \label{eq:genproblem}
h_t(\bm a)\triangleq f_t(\bm a)+ \Omega_t (\bm a),
\end{equation}
where
\begin{myitemize}
	\myitem $f_t: \mathcal X\rightarrow \mathbb{R}$ is a general convex loss function, and
	\myitem $\Omega_t: \mathcal X\rightarrow \mathbb{R}$ is a convex regularization function,
	\myitem with $\mathcal X$ being a convex set. 
\end{myitemize}
Note that the function $\Omega_t(\cdot)$ can vary with time, however, in this work, it will remain constant.

Given such a sequence of functions, the online learning setting requests to generate, at each time $t$, a hypothesis or estimate $\bm a[t]$, given the previous functions $\{h_\tau\}_{\tau=0}^{t-1}$. The quality of the proposed estimate $\bm a[t]$ will be assessed by $h_t(\bm a[t])$. Since the estimate must be delivered before $h_t$ is made available, the possibility of generating good estimates is subject to certain assumptions on how much the sequence of optimal estimates (which is only known in hindsight) changes over time. In the context of this work, $\bm a[t]$ corresponds to the VAR parameters, and the online learning task corresponds to the tracking of the time-varying topologies, subject to the assumption that the topology changes slowly over time.

\myitem \cmt{Static regret}The performance metric usually considered in online learning algorithms for static problems is static regret, which compares the algorithm's performance with a constant hindsight solution. 
\myitem \cmt{static regret and time-varying solutions}Although online algorithms with sublinear regret \cite{shalev2011online}
can be applied in practice, the static regret is not an adequate metric for quantifying how well an algorithm infers time-varying models.  
To characterize the performance of online algorithms in time-varying scenarios, the dynamic regret (where the hindsight solution is time-varying) is given by \cite{zinkevich2003online}: 
\vspace{-2mm}
\begin{equation} \label{eq:defdynamicregret}
R_d[T] \triangleq \sum_{t=1}^{T}\big [ h_t(\bm a[t])- h_t(\bm a^*[t])\big  ],\vspace{-1mm}
\end{equation}
where
\begin{myitemize}%
	\myitem \cmt{estimate}$\bm a[t]$ is the estimate of the online algorithm 
	\myitem \cmt{optimal solution}and $\bm a^*[t]$ is the optimal solution\footnote{For simplicity of exposition, $h_t(\cdot)$ is usually assumed to have a unique minimizer, which is verified by the loss function presented in Sec. \ref{sec:tirso}.} at time $t$, given by $\bm a^*[t]\triangleq \arg \min _{\bm a} h_t(\bm a)$. 
\end{myitemize}
Next, we present an online algorithm to solve the composite problem in \eqref{eq:genproblem}. 
\end{myitemize}%
\cmt{Online Proximal Gradient algorithm}%
\begin{myitemize}%
\myitem \cmt{Introduction to proximal algorithms}Composite problems can be efficiently solved via proximal methods \cite{parikh2014proximal,beck2017}, which exploits the proximity operator.
\begin{myitemize}%
	\myitem \cmt{Prox operator}The proximity (prox) operator of a scaled function $\eta\Psi$ at point $\bm v$ is defined by \cite {parikh2014proximal}:
	\begin{equation} \label {eq:defproximaloperator} 
	\textbf{prox}_{\Psi}^\eta(\bm v)\triangleq \underset{\bm x \in \text{dom }\Psi}{\arg\min}\left [\Psi(\bm x)+\frac{1}{2\eta}\left \lVert \bm x-\bm v\right \rVert_2^2\right],
	\end{equation}
	where
	\begin{myitemize}%
		\myitem \cmt{proximal term}$\Psi(\cdot)$ is minimized together with a quadratic proximal term, making the objective strongly convex.
		\myitem \cmt{Interpretations}The prox operator of a function at point $\bm v$ can be interpreted as minimizing the function while being close to $\bm v$, and 
		\myitem \cmt{regularization parameter}the parameter $\eta$ controls the trade-off between the two objectives. 
	\end{myitemize}
	\myitem \cmt{why proximal methods and different proximal algorithms}
	Proximal algorithms are used to solve composite problems, and they exhibit good convergence guarantees. 
	
	\myitem \cmt{Proximal gradient algorithm}An algorithm for solving composite problems is proximal gradient descent (PGD) \cite{parikh2014proximal}. 
	Until convergence, at each iteration, a gradient descent step is performed on the differentiable component of the objective and then the prox operator of the non-differentiable function at the resultant vector is performed. 
	\myitem \cmt{Online Proximal gradient algorithm}In its online version, namely proximal OGD, only one iteration of the proximal gradient is performed at each time instant based on the available data sample, instead of running until convergence.
	\myitem \cmt{Inexact online proximal gradient descent}In cases where the full information about the cost function is not available, \mbox{IP-OGD} \cite{dixit2019onlineproximal} assumes that an inexact gradient is available and the analysis of the algorithm includes the error between the true gradient and the available inexact gradient. The \mbox{IP-OGD} algorithm enjoys solid performance guarantees for tracking time-varying parameters.
\end{myitemize} 
\end{myitemize} 

\section{Deriving an approximate loss function}\label{sec:approximateLoss}

The expressions in the previous section [cf. \eqref{eq:dynamic_cost_1}] represent the problem of joint estimation and reconstruction from a rather ideal point of view because, even though the optimal policy would allow the best possible tracking, finding such a policy is nearly intractable. Fortunately, adding a few simple assumptions can give rise to a composite objective problem that can be solved using the approach described in Sec. \ref{sec:composite}.

Notice that at time $t$, considering the underlying observation noise and random missing data, the previous $P$ reconstructed samples, $\{ \hbm y[\tau]\}_{\tau = t-P}^{t-1}$, can be considered as realizations of random variables. This allows replacing the deterministic cost function $c_t(\cdot)$ for the batch formulation by the following \emph{random} cost function for the online problem:
\vspace{-3mm}
\begin{multline}
\label{eq:function_c}
C_t\left(
 \bm y[t],\left \{\bm A_p^{(t)} \right \}_{p=1}^{P}
 \right)
 =\\
 \dLoss_t \left (\{ \hbm y[\tau]\}_{\tau = t-P}^{t-1}, \bm y[t], \left \{\bm A_p^{(t)} \right \}_{p=1}^{P} \right )
+
\Omega\left(\left \{\bm A_p^{(t)}\right \}_{p=1}^{P}\right) 
\\
+ \beta  \sum_{p=1}^{P}\lVert \bm A_p^{(t)} -  \hbm A_p^{(t-1)} \rVert_{\text F}^2
,
\end{multline}
which is jointly convex in its arguments, and where $\{ \hbm A_p^{(t-1)} \}_{p=1}^{P}$ have been previously estimated at time $t-1$.
Notice that, if $\{ \hbm y[\tau]\}_{\tau = t-P}^{t-1}$ and $\tbm y[t]$ were equal to the true (yet unobservable) signals $\{ \bm y[\tau]\}_{\tau = t-P}^{t}$, this setting would be the same that is dealt with in \cite{zaman2019online}, by direct application of proximal OGD. \textcolor{black}{The setting here is more challenging because it involves a joint minimization over the estimated signals and the VAR model parameters.} Moreover, since the aforementioned signal estimates are inexact versions of the true signals, in the present work we will use the \mbox{IP-OGD} framework discussed in \cite{dixit2019onlineproximal} to analyze the regret of the resulting algorithm. Before proceeding to the formulation of the online algorithm, two remarks are in order. \\
{\bf Remark 2:} The cost function takes the signal estimate and the VAR parameters. It is assumed that the VAR parameters change smoothly with time, but we cannot assume that the signals vary smoothly with time. Recall that in each proximal OGD iteration, a minimization is solved involving a first-order approximation of the loss $\dLoss_t$, the regularizer $\Omega$ (not linearized), and a proximal term that ensures that the variable estimated at time $t$ is close in norm to the estimate at time $t-1$. This proximal term should involve $\{\bm A_p^{(t)}  \}_{p=1}^{P}$, but not $\bm y[t]$.
\\
{\bf Remark 3:} As a consequence of the variable decoupling introduced in \eqref{eq:function_c}, 
$C_t(\cdot)$ becomes separable across nodes.
\textcolor{black}{Thanks to this decomposability, the proposed algorithms can process the inputs of each node separately, but it does not completely remove the coupling among nodes in the online estimation process, as the inferred topology is used to reconstruct signals using the neighboring nodes.}

\textcolor{black}{The remainder of this section discusses the reformulation of the joint optimization over $\{\bm A_p^{(t)}  \}_{p=1}^{P}$ and $\bm y[t]$ into an optimization only over $\{\bm A_p^{(t)}  \}_{p=1}^{P}$. Upon application of IP-OGD, this will yield a proximal step involving only $\{\bm A_p^{(t)}  \}_{p=1}^{P}$, but the associated gradient is calculated differently to implicitly solve over $\bm y[t]$.
The aforementioned reformulation can be done as follows. 
}Note first that the joint minimization can be split into first minimizing over $\bm y[t]$ and then over $\{\bm A_p^{(t)}  \}_{p=1}^{P}$. The first minimization admits a closed form, which is convex in $\{\bm A_p^{(t)}  \}_{p=1}^{P}$. 
Specifically, we can write:
\newcommand{\rLoss}{\mathcal{L}}
\vspace{-1mm}
\begin{equation}
\min_{\! \bm y[t],\!\left \{\bm A_p^{(t)}\right \}_{p=1}^{P}}\!\!\!\!\! C_t\left(
 \bm y[t],\left \{\bm A_p^{(t)} \right \}_{p=1}^{P}
 \right)
  \!\!= \!\!\!\! \min_{ \left \{\bm A_p^{(t)}\right\}_{p=1}^{P}} \!\!\!\!\!
  \rLoss_t\left(\left \{\bm A_p^{(t)} \right \}_{p=1}^{P}\right), 
\end{equation}
\vspace{-1mm}
where
\begin{equation} \label{eq:randomloss}
\rLoss_t\left(\left \{\bm A_p^{(t)} \right \}_{p=1}^{P}\!\right) 
\!\!\define 
\min_{\bm y[t]} \dLoss_t \! \left( 
\{ \hbm y[\tau]\}_{\tau = t-P}^{t-1}, \bm y[t], \left \{\bm A_p^{(t)} \right \}_{p=1}^{P}
\right).
\end{equation}
\textcolor{black}{The main difficulty at this point is that the loss function in \eqref{eq:randomloss} is defined as the output of a minimization operator, and in order to apply IP-OGD, one needs its gradient in closed form. Fortunately, it is possible to derive the analytical minimization of \eqref{eq:randomloss}, which is shown in Sec. \ref{ss:analytical_reconstruct}}. Once a closed form is available for $\rLoss_t$, \mbox{IP-OGD} can be applied. 
The inexactness comes from the previously estimated (reconstructed) $\{ \hbm y[\tau]\}_{\tau = t-P}^{t-1}$. Recall that we model such estimates as random variables from the point of view of the agent that estimates $ \{\bm A_p^{(t)} \}_{p=1}^{P}$ at time $t$. That is what makes $\rLoss_t$ a random function, more specifically an inexact version of the ``true" loss function, which would be given by
\vspace{-1mm}
\begin{equation}
\label{eq:trueloss}
\rLoss^{\text{ true}}_t\left({\text{$\{\bm A_p^{(t)} \}_{p=1}^{P}$}}\right)\! \define \! \min_{\bm y[t]} \dLoss_t \! \left( 
\{ \bm y[\tau]\}_{\tau = t-P}^{t-1}, \bm y[t],  \{\bm A_p^{(t)}  \}_{p=1}^{P}
\right),
\nonumber
\end{equation}
but is unavailable because the true signal values $\{ \bm y[\tau]\}_{\tau = t-P}^{t-1}$ would be needed to evaluate it.
\\
Note that the loss function in \eqref{eq:randomloss} is separable across nodes:
\vspace{-2mm}
\begin{multline} \label{eq:separableloss}
	\rLoss_t\left(\left \{\bm A_p^{(t)} \right \}_{p=1}^{P}\right) 
	= \sum_{n=1}^N \rLoss_t^{(n)}(\bm a_n[t])\\
	=	\sum_{n=1}^N \min_{y_n[t]}  \ell_t^{(n)}(\hbm g[t],  y_n[t],\bm a_n[t]) ,
\end{multline}
\vspace{-2mm}
where
\vspace{-1mm}
\begin{multline}
	\ell_t^{(n)}(\hbm g[t], y_n[t],\bm a_n[t])\define \\
	 \frac{1}{2}\left(
	(y_n[t] - \hbm g[t]^\top \bm a_n[t])^2 + \frac{\nu \, m_n[t]}{|\mathcal M_t|} (y_n[t] - \tilde y_n[t])^2
	\right),
 \vspace{-1mm}
\end{multline}
\begin{equation}
		\hbm g[t] \define  \textrm{vec}\Big (\big [\hbm y[t-1], \ldots, \hbm y [t -P] \big ]^\top \Big ), \label{eq:defOfgHat1}
\end{equation}
\begin{equation} \label{eq:lossJSTISO}
	\rLoss_t^{(n)}(\bm a_n[t])
	\define 
	\min_{y_n[t]} \ell_t^{(n)}(\hbm g[t],  y_n[t],\bm a_n[t]).
\end{equation}
To arrive at the loss function, the minimizer (signal reconstruction) will be derived; then, a closed-form expression for $\rLoss_t^{(n)}$ will be obtained.
\vspace{-2mm}
\subsection{Signal reconstruction and loss function in closed form} \label {ss:analytical_reconstruct}
\cmt{Signal estimation problem} We discuss here the (sub)problem of estimating the signal from a noisy observation vector with missing values given a (fixed) topology. The resulting estimator is a convex combination of the signal prediction via the VAR process and the values present in the observation vector.  
\begin{myitemize}%
	\myitem \cmt{Online problem of signal estimation}More formally, the reconstruction subproblem consists in estimating $\bm y[t]$ given $\tbm y[t]$, $\bm m[t]$,  $\hbm g[t]$,  
and 
$\{ \mathbf{  A}_p^{(t)}\}_{p=1}^P$. Notice from \eqref{eq:def_dloss} that $\tbm y[t]$ and $\bm m[t]$ are implicit in the definition of $\ell_t(\cdot)$:  
\vspace{-2mm}
	\begin{equation} \label{eq:onlinesignalestimationproblem1}
	\hbm y[t] = 
	\arg \min_{\bm y[t]} \ell_t (\hbm g[t], \bm y[t],\bm a_n[t]).
	\end{equation}
The solution for the $n$-th entry of $\hbm y[t] $ is $\hat y_n[t] = \arg \min_{y_n[t]} \ell_t^{(n)}(\hbm g[t], y_n[t],\bm a_n[t])$, which has a closed form given by
\vspace{-2mm}
\begin{equation}\label{eq:nthsolutionsignalestimation}
	\hat y_n[t] = \left(1-U_n[t]\right) \hbm g[t]^\top \bm a_n[t] + U_n[t] \tilde y_n[t],
\end{equation} 
where 
\vspace{-2mm}
\begin{equation}
\label{eq:u}
	U_n[t] \define \frac{\nu}{|\mathcal M_t|+\nu }  m_n[t].
\end{equation}
Observe that $U_n[t]\in [0, \nu/(1+\nu)]$ holds $\forall t, n$.

\end{myitemize}
After substituting 
\eqref{eq:nthsolutionsignalestimation} into \eqref{eq:lossJSTISO} and simplifying, the loss function can be expressed as
\vspace{-2mm}
\begin{equation}
	\rLoss_t^{(n)}(\bm a_n[t]) = \frac{1}{2}U_n[t]
 (\tilde y_n[t] -\hbm g^\top[t] \bm a_n[t])^2,
\end{equation}
and it will be used in Sec. \ref{ss:derive_jstiso} to derive the \mbox{IP-OGD} iterates. Proximal OGD involves linearizing part of the objective, here $\rLoss_t^{(n)}(\cdot)$, which requires the gradient. The latter 
is given by 
\begin{eqnarray}\label{eq:gradfTISO}
	\nabla
	\rLoss_t^{(n)}(\bm a_n[t]) \hspace{-0.5mm} = \hspace{-0.5mm}  U_n[t] \left(\hbm g[t] \hbm g^\top [t] \bm a_n[t] -\tilde y_n[t]  \hbm g[t]\right).
\end{eqnarray}
\vspace{-5mm}
\subsection{Application of \mbox{IP-OGD} to Joint Signal and Topology Estimation} \label{ss:derive_jstiso}
\begin{myitemize}
\myitem \cmt{Application of Inexact Online Proximal Gradient Descent to the problem}
The gradient defined in \eqref{eq:gradfTISO} 
depends on $\hbm g[t]$, which is conformed using the estimates $\{\hbm y[t-p]\}_{p = 1}^{P}$, which in turn will generally differ from the true signals. This is translated into an error in the gradient and this is why IP-OGD is advocated here.
\par  
\myitem \cmt{Deriving the algorithm}
\cmt{JSTIRSO as OPGD}Let $f_t^{(n)}$ be a general loss function\nextver{inaccessible}, and let $\mathcal{F}_t^{(n)}$ be a random function that is an inexact version of $f_t^{(n)}$. Using $\mathcal{F}_t^{(n)}$ and $\Omega^{(n)}(\bm a_n)$ in \eqref{eq:genproblem}, with a constant step size $\alpha$, the IP-OGD iteration is:
\begin{equation}
\bm a_n[t] = \textbf{prox}_{\Omega^{(n)}}^{\alpha} \left ( \bm a_n[t-1] - \alpha \nabla \mathcal{F}_t^{(n)}(\bm a_n[t-1])   \right ).
\end{equation}

Let $\bm a_n^{\text{f}} [t] \triangleq  \bm a_n[t-1] - \alpha  \nabla \mathcal{F}_t^{(n)}(\bm a_n[t-1])$,  and 
\begin{equation} \label{eq:af_splitting}
\bm a^\text{f}_{n}[t]= [(\bm a^\text{f}_{n,1}[t])^\top, \ldots, (\bm a^\text{f}_{n,N}[t])^\top ]^\top,
\end{equation}
which enables us to write the above update expression as
\begin{align*}
\bm a_n[t]& =  \textbf{prox}_{ \Omega^{(n)}}^{\alpha} \left ( \bm a_n^{\text{f}} [t]   \right )\\
& = \underset{\bm z_n }{\arg\min} \left ( \Omega ^{(n)} (\bm z_n) + \frac{1}{2\alpha} \left \lVert \bm z_n - \bm a_n^{\text{f}} [t] \right \rVert_2^2 \right ).
\end{align*}
Using the regularizing function $\Omega^{(n)}(\bm a_n) \define \lambda \sum_{n'=1}^{N} \mathds 1 \{ n \neq n'\}  \left\lVert \bm a_{n,n'}  \right\rVert_2 
$, [cf. \eqref{eq:omega_definition} that $\Omega = \sum_{n=1}^{N} \Omega^{(n)}$], the update yields
 \begin{align*}
\bm a_n[t]& = \underset{\{\bm z_{n,n'}\}_{n'=1}^N  }{\arg\min} \Bigg ( \lambda \sum_{n'=1}^{N} \mathds 1 \{ n \neq n'\}  \left\lVert \bm z_{n,n'}  \right\rVert_2\\
& \quad + \frac{1}{2\alpha} \sum_{n'=1}^{N} \left \lVert \bm z_{n,n'} - \bm a_{n,n'}^{\text{f}} [t] \right \rVert_2^2 \Bigg ).
\end{align*}
which is separable across $n^{\prime}$ and the solution to the $n'$-th subproblem is given by the group soft-thresholding:
\vspace{-2mm}
\begin{align}
\bm a_{n,n'}[t] &= \underset{\bm z_{n,n'} }{\arg\min} \Bigg [ \mathds 1 \{ n \neq n'\}  \left\lVert \bm z_{n,n'}  \right \rVert_2 \nonumber \\ &\quad \quad \quad \quad \quad + \frac{1}{2\alpha \lambda} \left \lVert \bm z_{n,n'} - \bm a_{n,n'}^{\text{f}} [t] \right \rVert_2^2 \Bigg] \nonumber \\
&= \bm a^\text{f}_{n,n'}[t ]\left [1-\frac{\alpha \lambda  \mathds 1 \{ n \neq n'\}}{\left \lVert \bm a^\text{f}_{n,n'}[t ] \right\rVert_2}\right]_+, \label {eq:updateproxonline}
\end{align}
(recall that $\bm a^\text{f}_{n,n'}[t]$ is a subvector of $\bm a^\text{f}_{n}[t]$ as defined in \eqref{eq:af_splitting}).
\myitem \cmt{Online algorithm for tracking time-varying topologies}The algorithm JSTISO, which is intended at minimizing $C_t(\cdot, \cdot)$ in \eqref{eq:function_c}, is obtained when $\mathcal{F}_t^{(n)}$ is set to be $\rLoss_t^{(n)}$. All required steps are summarized in Procedure \ref{alg:JSTISO}. It only requires $\mathcal{O}(N^2 P)$ memory entries to store the previous $P$ reconstructed samples, and each update requires $\mathcal{O}(N^2 P)$ arithmetic operations, which is in the same order as the number of parameters to be estimated.

{\bf Remark 4:} For those time instants and nodes where an entry is missing, $\nabla_{\bm a_n[t]}\rLoss_t^{(n)} = \bm 0_{NP}$, but \textbf{Procedure \ref{alg:JSTISO}} applies the soft-thresholding operator \eqref{eq:updateproxonline} 
to the corresponding coefficients. While this may seem counter-intuitive, the shrinking is justified by the model at hand. The time-varying parameters are modeled as a random walk whose innovations are compound by a) a Gaussian distributed term, plus b) a term that attracts the VAR parameters towards $\bm 0$ 
for sparsity.
The term a) justifies the Frobenius norm in \eqref{eq:jointbatchproblem} and the term b) justifies the presence of $\Omega$ in the same equation.

\textcolor{black}{{\bf Remark 5:} \cmt{It should be possible to recover TISO} In addition to the different input to the method, the proposed algorithm in \textbf{Procedure 1} differs from TISO \cite{zaman2019online} in lines 5, 6, and 12. Steps 5 and 6 correspond to the computation of the gradient while step 12 corresponds to the estimation of the signal.
}

\begin{algorithm}
\caption{Tracking time-varying topologies with missing data via JSTISO}
\textbf{Input parameters:} $P, \lambda, \alpha, \nu$ \\
\textbf{Initialization:} $   \{\hbm y[\tau]\}_{\tau =0}^{P-1}, \{{\bm a}_n[P-1]
\}_{n=1}^{N} $ 
\begin{algorithmic}[1] 
	\For {$t=P,P+1, \ldots $}
	\State {Receive observation $\tbm y[t]$ and masking vector $\bm m[t]$} 
	\State {Obtain $\hbm g[t] $ from $\{\hbm y[t-p]\}_{p=1}^P$  via \eqref{eq:defOfgHat1} }
	\For {$n=1, \ldots,N $} 
    \State{\textcolor{black}{$U_n[t] = \frac{\nu}{|\mathcal M_t|+\nu }  m_n[t]$}}
	\textcolor{black}{\State $\hbm v_n[t]=U_n[t] \left(\hbm g[t] \hbm g^\top [t] \bm a_n[t-1] -\tilde y_n[t]  \hbm g[t]\right)$ }
	\State $\bm a^\text{f}_{n} [t]= \bm a_n[t-1] - \alpha \textcolor{black}{\hbm v_n[t]}$
	\For {$n'=1,2, \ldots ,N $ } 
	\State{Compute ${\bm a}_{n,n'}[t]$ via \eqref{eq:updateproxonline}
	}
	\EndFor 
	\State \textbf{end for}
	\State {$ \bm a_n[t]=\left [{\bm a}_{n,1}^\top [t], \ldots, {\bm a}_{n,N}^\top [t]\right ]^ \top$ }
    \State \textcolor{black}{$\hat y_n[t] = \left(1-U_n[t]\right) \hbm g[t]^\top \bm a_n[t] + U_n[t] \tilde y_n[t]$}
	\EndFor
	\State \textbf{end for}
	\State {Output $\left \{  {\bm a}_n[t]\right  \}_{n=1}^N, \hbm y[t]$}
	\EndFor	
	\State \textbf{end for}				
\end{algorithmic}\label{alg:JSTISO}
\end{algorithm}
\end{myitemize}

\section{An Alternative Loss Function for Improved Tracking} \label{sec:tirso}
The loss function in the previous approach is an instantaneous loss, which only depends on the current sample. Albeit it has low computational complexity per iteration and can be sufficient for online estimation of a static VAR model, it is sensitive to noise and input variability, and thus it may be not suitable for a time-varying model. In \cite{zaman2019online}, a running average loss function is designed drawing inspiration from the relation between least mean squares (LMS) and recursive least squares (RLS) to improve the tracking capabilities of TISO. In this paper, similar steps will lead to a second approach, where a running average loss function is adopted, which depends on the past reconstructed signal values. In the second approach, the loss function is set as: 
\vspace{-3mm}
\begin{multline} \label{eq:JSTIRSOlossfunction}
 \tilde \dLoss_t\left( 
\{ \hbm y[\tau]\}_{\tau = 0}^{t-1}, \bm y[t],  \left \{\bm A_p^{(t)} \right \}_{p=1}^{P}
 \right )=\\
\frac{1}{2} 
\left \lVert \bm y[t]\!\!-\!\!\sum_{p=1}^{P} \!\!\bm A_p^{(t)} \, \hat {\bm y}[t-p]\right \rVert_2^2 + \frac{\nu}{2 |\mathcal M_t|} \left \lVert \tbm y[t]-\bm m[t] \odot \bm y[t]\right \rVert_2^2 \\
 +\frac{1}{2} \sum_{\tau=P}^{t-1}\gamma^{t-\tau}\left \lVert \hat{\bm y}[\tau]-\sum_{p=1}^{P} \bm A_p^{(t)} \, \hat {\bm y}[\tau-p]\right \rVert_2^2,
\end{multline}
where $\gamma$ is a user-selected forgetting factor that controls the weight of past (reconstructed) samples of $\bm y[t]$. The modeling principles in the previous section (treating the previously reconstructed samples as a random variable, and minimizing over $\bm y[t]$) are applied to the alternative deterministic loss $\tilde \dLoss_t$, enabling to define the random loss function $\tilde \rLoss_t$ as 
\begin{equation} \label{eq:randomlossJSTIRSO}
\tilde \rLoss_t\left(\left \{\bm A_p^{(t)} \right \}_{p=1}^{P}\right) 
\define 
\min_{\bm y[t]} \tilde \dLoss_t\left( 
\{ \hbm y[\tau]\}_{\tau = 0}^{t-1}, \bm y[t], \left \{\bm A_p^{(t)} \right \}_{p=1}^{P}
\right),
\end{equation}
which can be rewritten in terms of 
$\ell_t$ as:
\begin{multline} \label{eq:randomlossJSTIRSO1}
\! \!\!\!	\tilde \rLoss_t\left(\left \{\bm A_p^{(t)} \right \}_{p=1}^{P}\right) 
\!\!	= 
	\min_{\bm y[t]}  \dLoss_t\left( 
	\{ \hbm y[\tau]\}_{\tau = t-P}^{t-1}, \bm y[t], \left \{\bm A_p^{(t)} \right \}_{p=1}^{P}	\right) \\ + \frac{1}{2}\sum_{\tau=P}^{t-1}\gamma^{t-\tau}\left \lVert \hat{\bm y}[\tau]-\sum_{p=1}^{P} \bm A_p^{(t)} \, \hat {\bm y}[\tau-p]\right \rVert_2^2.
\end{multline}
Regarding the signal reconstruction, 
the minimizer of \eqref{eq:randomlossJSTIRSO} is:
\begin{align}
\hbm y[t] &= \underset{  \bm y[t] }{\arg \min}~ \tilde \dLoss_t\left( 
\{ \hbm y[\tau]\}_{\tau = 0}^{t-1}, \bm y[t], \left \{\bm A_p^{(t)} \right \}_{p=1}^{P} \right )\nonumber \\
&= \underset{  \bm y[t] }{\arg \min} \frac{1}{2} 
\left \lVert \bm y[t]-\sum_{p=1}^{P} \bm A_p^{(t)} \, \hat {\bm y}[t-p]\right \rVert_2^2 \nonumber \\
& \quad + \frac{\nu}{2 |\mathcal M_t|} \left \lVert \tbm y[t]-\bm m[t] \odot \bm y[t]\right \rVert_2^2. \label{eq:signalreconsJSTIRSO}
\end{align}
Observe that \eqref{eq:signalreconsJSTIRSO} coincides with the reconstruction problem in \eqref{eq:onlinesignalestimationproblem1}
and, therefore, its solution
is given by \eqref{eq:nthsolutionsignalestimation}.

Next, to derive the closed-form solution for $\tilde \rLoss_t$, we substitute the closed-form expression of $\hbm y[t]$ from \eqref{eq:nthsolutionsignalestimation} into \eqref{eq:randomlossJSTIRSO}:
\begin{align}
&\tilde \rLoss_t \left (\left \{\bm A_p^{(t)} \right \}_{p=1}^{P} \right )= \nonumber \\
&\frac{1}{2} \sum_{n=1}^N \Bigg [ U_n[t]
(\tilde y_n[t] -\hbm g^\top[t] \bm a_n[t])^2 \Bigg ]
 +\frac{1}{2} 
\sum_{n=1}^N \Big( \sum_{\tau=P}^{t-1}\gamma^{t-\tau} \hat y_n^2[\tau] \nonumber \\
&\quad  +\gamma \bm a_n^\top[t] \hbm \Phi[t-1] \bm a_n[t] - 2 \gamma \hbm r_n^\top[t-1] \bm a_n[t] \Big), \label{eq:lossJSTIRSO}
\end{align}
\nextver{try to respace this eq}
where
\vspace{-5mm}
\begin{subequations}
	\begin{align}
	\hbm \Phi[t] &\define \sum_{\tau=P}^{t} \gamma^{t-\tau}\hbm g[\tau] \hbm g^\top[\tau], \\
	\hbm r_n[t] & \define \sum_{\tau=P}^{t} \gamma^{t-\tau} \hat y_n[\tau] \hbm g[\tau].
	\end{align}
\end{subequations}
The variables above can be efficiently computed via recursive expressions.\footnote{The recursive expressions are presented in lines 4 and 6 in Procedure \ref{alg:JSTIRSO}.}
Note that
$\tilde{\rLoss}_t$ is also separable across nodes, i.e., 
\vspace{-2mm}
\begin{equation}
\tilde \rLoss_t (\cdot)= \sum_{n=1}^N \tilde\rLoss_t^{(n)}(\cdot),
\end{equation} 
\vspace{-2mm}
where 
\vspace{-2mm}
\begin{multline}
	\tilde\rLoss_t^{(n)} (\bm a_n) \define \rLoss_t^{(n)}(\bm a_n)+ \sum_{\tau=P}^{t-1}\gamma^{t-\tau} \hat y_n^2[\tau] +\gamma \bm a_n^\top\hbm \Phi[t-1] \bm a_n \\
	- 2 \gamma \hbm r_n^\top[t-1] \bm a_n.
\end{multline}
\begin{figure}
    \centering
    \includegraphics[width=\columnwidth]{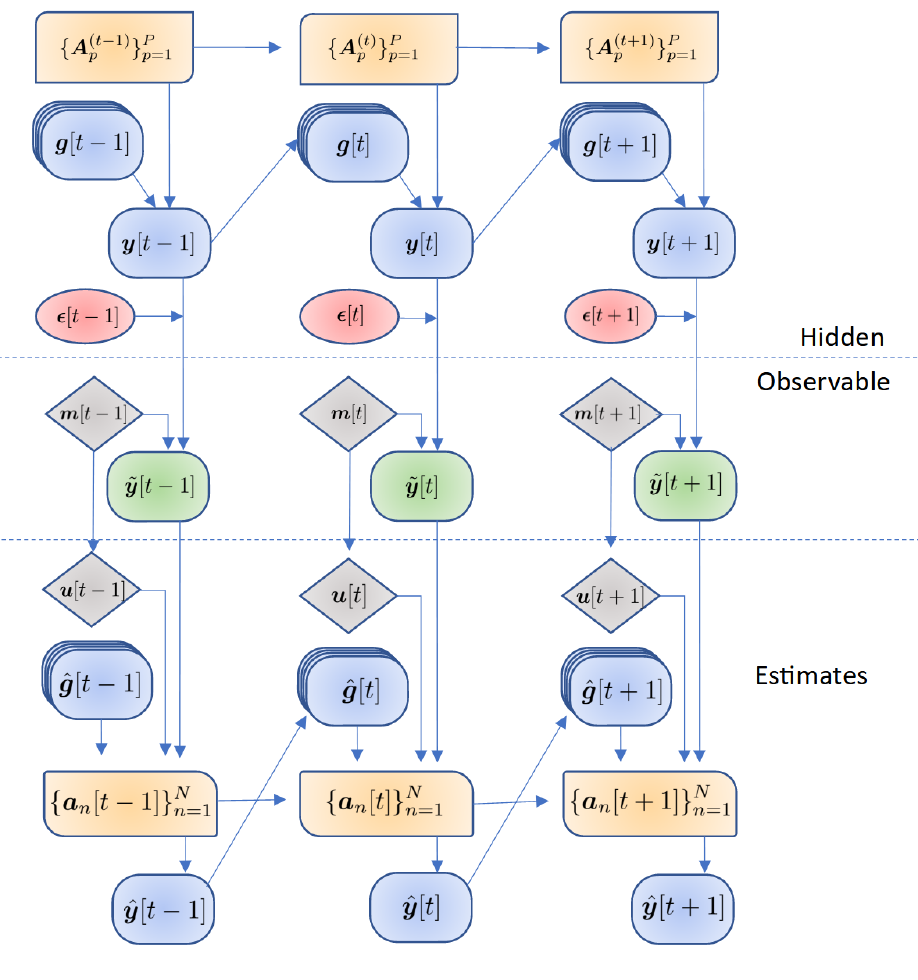}
    \caption{A simplified flow diagram of JSTISO/JSTIRSO.}
    \label{fig:Schematic}
\end{figure}

The algorithm JSTIRSO 
is obtained when $\mathcal{F}_t^{(n)}$ is set to be $\tilde \rLoss_t^{(n)}$, following similar steps to those in Sec. \ref{ss:derive_jstiso}. 
The gradient of $\tilde \rLoss_t^{(n)}$ w.r.t. $\bm a_n[t]$ is given by
\begin{multline}\label{eq:gradfJSTIRSO}
\nabla\tilde \rLoss_t^{(n)} (\bm a_n[t]) = 
U_n[t] \left(\hbm g[t] \hbm g^\top [t] \bm a_n[t] -\tilde y_n[t]  \hbm g[t]\right) \\
+\gamma \hbm \Phi[t-1] \bm a_n[t] - \gamma \hbm r_n[t-1].
\end{multline}	
All steps are summarized in \textbf{Procedure~\ref{alg:JSTIRSO}}.
The computational complexity of JSTIRSO is dominated by step 8 of \textbf{Procedure~2}, which is $\mathcal{O}(N^3P^2)$ operations per $t$.

The initial values for $\hbm \Phi[P-1]$ and $\hbm r_n[P-1]$ can be set depending on available prior information; if no such information is available, one can choose a small $\sigma$ and set $\hbm \Phi[P-1] = \sigma^2 \bm I$, and $\hbm r_n[P-1] = \bm 0, \forall ~ n$. 
A schematic diagram illustrating the variables involved in the generation, (partial) observation of signal entries, and estimation via JSTISO/JSTIRSO is given in Fig. 1.
The figure shows how the present estimate of topology parameters and the present estimate of the signal are dependent on the previous estimates of topology and signal in a sequential manner.
\begin{algorithm}
\caption{Tracking time-varying topologies with missing data via JSTIRSO}
\textbf{Input:} $P, \lambda, \alpha, \nu, \gamma,  \sigma^2 $ \\
\begin{tabular}{rl}
\hspace{-2mm}\textbf{Initialization:} \hspace{-4mm}
& 
$\{\hbm y[\tau]\}_{\tau =0}^{P-1},\{{\tbm a}_n[P-1], \hbm r_n[P-1]\}_{n=1}^{N}, $
\\ 
&
$\hbm \Phi[P-1]$\\ 
\end{tabular}
\begin{algorithmic}[1] 
	\For {$t=P,P+1, \ldots $} 
	\State {Receive observation $\tbm y[t]$ and masking vector $\bm m[t]$} 
	\State {Obtain $\hbm g[t] $ from $\{\hbm y[t-p]\}_{p=1}^P$  via \eqref{eq:defOfgHat1} }
	\State $\hbm \Phi[t]= \gamma \, \hbm \Phi[t-1]+\hbm g[t]  \hbm g^\top[t]$ 
	\For {$n=1, \ldots,N $} 
	\State $\hbm r_n[t]= \gamma \, \hbm r_n[t-1]+ \tilde y_n[t]\, \hbm g[t]$ 
    \State{\textcolor{black}{$U_n[t] = \frac{\nu}{|\mathcal M_t|+\nu }  m_n[t]$}}
    \hspace{5mm}\textcolor{black}{\State{
    \begin{align*}
        \hat{\tilde{\bm v}}_n[t]=&U_n[t] \left(\hbm g[t] \hbm g^\top [t] \tbm a_n[t-1] -\tilde y_n[t]  \hbm g[t]\right)+\\
    &\gamma \hbm \Phi[t-1] \tbm a_n[t-1] - \gamma \hbm r_n[t-1]
    \end{align*}\vspace{-5mm}
    }}
	\State 
	$\tbm a^\text{f}_{n} [t]= \tbm a_n[t-1] - \alpha \textcolor{black}{ \hat{\tilde{\bm v}}_n[t]}$ [cf. \eqref{eq:gradfJSTIRSO}]	
	\For {$n'=1,2, \ldots ,N $ } 
	\State {${\tbm a}_{n,n'}[t]=\tbm a^\text{f}_{n,n'}[t ]\left [1-\frac{\alpha \lambda ~ \mathds 1 \{ n \neq n'\}}{\left \lVert \tbm a^\text{f}_{n,n'}[t ] \right\rVert_2}\right]_+$ }
	\EndFor 
	\State \textbf{end for}
	\State {$ \tbm a_n[t]=\left [{\tbm a}_{n,1}^\top [t], \ldots, {\tbm a}_{n,N}^\top [t]\right ]^ \top$ }
    \State \textcolor{black}{$\hat y_n[t] = \left(1-U_n[t]\right) \hbm g[t]^\top \tbm a_n[t] + U_n[t] \tilde y_n[t]$}
	\EndFor
	\State \textbf{end for}
	\State {Output $\left \{  {\tbm a}_n[t]\right  \}_{n=1}^N, \hbm y[t]$}
	\EndFor
	\State \textbf{end for}					
\end{algorithmic}\label{alg:JSTIRSO}
\end{algorithm}

\section{Performance analysis}
\label{sec:analysis}
To analyze the performance of JSTIRSO, we present analytical results in this section. First, the assumptions considered in the analysis are stated and then, two Lemmas followed by the main theorem about the dynamic regret bound of JSTIRSO are presented. Moreover, a third lemma stating a bound on the error in the gradient is presented and discussed. Finally, a corollary with a simpler dynamic regret bound is presented. 

To quantify the inexactness in our algorithm, we need to define the following quantities: 
\begin{subequations}
	\begin{align}	
		\bm g[t]    & \define  \textrm{vec}\left (\big [\bm y[t-1], \ldots, \bm y [t -P] \big ]^\top \right ), \label{eq:defTrueg} \\
		\bm \Phi[t] &\define \sum_{\tau=P}^{t} \gamma^{t-\tau}\bm g[\tau] \bm g^\top[\tau], \label{eq:defTruePhi} \\
		\bm r_n[t] & \define \sum_{\tau=P}^{t} \gamma^{t-\tau} \tilde y_n[\tau] \bm g[\tau], \label{eq:defTruer}
	\end{align}
\end{subequations}
which can be respectively thought as the true versions of $\hbm g[t]$, $\hbm \Phi[t]$ and $\hbm r_n[t]$.
\begin{myitemize} 
\myitem \cmt{Assumptions}The following assumptions will be considered for the characterization of JSTIRSO:
\begin{enumerate}[{A}1.]
	\item {\textit{Bounded samples:} There exists $\energybound\!>0\!$
		such that $|y_n[t]|^2 \leq \energybound$, $|\hat y_n[t]|^2 \leq \energybound$, and $|\tilde y_n[t]|^2 \leq \energybound ~\forall \, n, t$. } \label{as:boundedprocess}
	\item {\textit{Bounded minimum eigenvalue of $\bm \Phi[t]$ and $\hbm \Phi[t]$:}
		There exists $\strongcvxparf~ >~0$ such that  $\lambda_{\mathrm{min}}(\bm \Phi[t])~\geq~ \strongcvxparf$ and $\lambda_{\mathrm{min}}(\hbm \Phi[t]) \geq \strongcvxparf , ~\forall \, t \geq P$.} \label{as:mineig}
	\item {\textit{Bounded maximum eigenvalue of $\bm \Phi[t]$ and $\hbm \Phi[t]$:}
		There exists  $L~>~0$ such that $\lambda_{\mathrm{max}}(\bm \Phi[t]) \leq L$ and $\lambda_{\mathrm{max}}(\hbm \Phi[t]) \leq L, ~\forall \, t \geq P$. }\label{as:maxeig}
	\item {\textit{Bounded errors in $ \bm g, \bm \Phi, \bm r_n$ due to noise, missing values:
		}
		\vspace{-4mm}
		\begin{subequations} \label{eqs:boundederrors}
		\begin{align}
		\label{eq:bg} \left \lVert  \hbm g[t]-\bm g[t] \right \rVert_2 & \leq B_{\bm g} \quad \forall ~t\\
		\label{eq:bphi} \lambda_{\mathrm{max}} \left ( \hbm \Phi[t]- \bm \Phi[t]\right) &\leq B_{\bm\Phi} \quad  \forall ~t\\
		\label{eq:br} \left \lVert  \hbm r_n[t]-\bm r_n[t] \right \rVert_2 & \leq B_{\bm r} \quad \forall~ n,t. 
		\end{align}
	\end{subequations}\label{as:errorbounds}}
\end{enumerate}%
\vspace{-4mm}
A1 entails no loss of generality since data are bounded in real-world applications. 
A2 holds in practice unless the data are redundant (meaning that some time series can be obtained as a linear combination of the others), that is, 
it will be satisfied for a sufficiently large number of samples. Thus, A2 is a reasonable assumption in real-world applications. 
A3 is fulfilled when the true signal values and their corresponding reconstructed values are bounded. 
A4 sets a limit on the magnitude of the error introduced in various quantities due to noise and missing values. 
A4 is satisfied when the noise and number of missing values are limited such that the errors in $ \bm g, \bm \Phi, \bm r_n$ after the signal estimation step are always bounded by the given constants.

The next results depend on the error in the gradient, i.e.,
\begin{equation}\label{eq:erroringrad}
\bm e^{(n)}[t] \define \nabla \tilde \rLoss_t^{(n)}\!(\bm a_n[t])-  \nabla
\tilde \rLoss_t^{(n) \text{true}} (\bm a_n[t]),
\end{equation}
where 
$
\tilde \rLoss_t^{(n) \text{true}} (\bm a_n[t]) \define \min_{ y_n[t]} \tilde \dLoss_t^{(n)} \left( 
\{ \bm y[\tau]\}_{\tau = 0}^{t-1},  y_n[t],\! \bm a_n[t]  \right) 
$
is the true (exact) gradient (where $\{ \bm y[\tau]\}_{\tau = 0}^{t-1}$ are the (unobservable) true signal values),
and 
$\nabla \tilde \rLoss_t^{(n)}(\bm a_n)$ 
is the inexact gradient defined in \eqref{eq:gradfJSTIRSO}. The latter is inexact due to the error in the reconstructed entries of $\hbm g$, and the error in $\hbm g$ comes in turn from the missing values and noisy observations.

\myitem \cmt{Regret bound}Dynamic regret analysis is generally expressed in terms of metrics that express how challenging tracking becomes, e.g., how fast the optimal parameters vary. In our specific case, the dynamic regret will be expressed in terms of the variation in consecutive optimal solutions (often referred to as path length \cite{zinkevich2003online}) and the error in the gradient \cite{dixit2019onlineproximal}. If we define $\tilde h_t^{(n)} \define \tilde \rLoss_t^{(n)} +\Omega^{(n)}$, and let $ \timevarhindsight[t] \define \arg \min _{\bm a_n} \tilde h_t^{(n)}(\bm a_n)$ be the time-varying optimal solution, the path length is given by 
\vspace{-2mm}
\begin{equation} \label {eq:pathlength}
W^{(n)}[T] \triangleq \sum_{t=P+1}^T \left \lVert \timevarhindsight[t] - \timevarhindsight[t-1] \right \rVert_2.
\end{equation}
Also, we define the cumulative (norm of the) gradient error as 
\vspace{-2mm}
\begin{equation} \label {eq:cumulativeerror} 
E^{(n)}[T] \triangleq \sum _{t=P}^T 
 \left \lVert \bm e ^{(n)}[t] \right \rVert_2 .
\end{equation} 
The dynamic regret for JSTIRSO for the $n$-th node is:
\begin{equation} \label{eq:defdynamicregretTIRSO}
\tilde R_d^{(n)}[T]\triangleq  \sum_{t=P}^{T}\big [ \tilde h_t^{(n)}(\tbm a_n[t])- \tilde h_t^{(n)}( \timevarhindsight[t])\big  ],
\end{equation}
where
\begin{myitemize}%
	\myitem \cmt{estimate}$\tbm a_n[t]$ is the JSTIRSO topology estimate. Next, we present two lemmas that will be instrumental in deriving the dynamic regret of JSTIRSO. 
\end{myitemize}
\begin{mylemma} \label{lemma:boundedgrad}
	Under assumptions A\ref{as:boundedprocess} and A\ref{as:maxeig}, we have
	\begin{multline} \label{eq:JSTIRSOgradBound}
		 \left \lVert \nabla \tilde \rLoss_t^{(n)} (\tbm a_n[t]) \right \rVert_2  \leq \\
\frac{\nu}{1+\nu} \left(PN \energybound+ 2 \sqrt{PN \energybound} B_{\bm g} + B_{\bm g}^2 +\gamma L\frac{1+\nu}{\nu}\right ) \times\\ \frac{1}{\strongcvxparf \gamma }\!\left (\!\! \frac{\nu}{1\!+\!\nu}  \sqrt{PN}\energybound \! +\! \frac{\sqrt{PN} \energybound}{1\!-\!\gamma }\!\!\right)\! +\! \left ( \!\frac{\nu}{1\!+\!\nu} \!+\!\frac{\gamma}{1\!-\!\gamma } \!\right)\!\! \sqrt{PN} \energybound \\
 \define  B_{\bm v}
	\end{multline}
\end{mylemma}  
\begin{proof}
	See Appendix A in the supplementary material.
\end{proof}
\begin{mylemma}\label{lemma:boundedsubgrad}
	All the subgradients of the regularization function $\nReg$ are bounded by  $\lambda \sqrt{N}$, i.e., $\lVert \bm u_t \rVert_2\leq \lambda \sqrt{N}$, where $ \bm u_t \in \partial  \Omega^{(n)}(\tbm a_n[t])$.
\end{mylemma}
\begin{proof}
	See the proof of Theorem 5 in \cite{zaman2019online}.
\end{proof}
Next, we present a bound on the dynamic regret of JSTIRSO.
\cmt{Theorem statement}%
\begin{theorem} \label{th:dynamicregretboundTIRSO}
	\begin{myitemize}%
		\myitem \cmt{assumptions}Under assumptions A\ref{as:boundedprocess}, A\ref{as:mineig},  A\ref{as:maxeig}, and A\ref{as:errorbounds}, 
		\myitem \cmt{JSTIRSO }let $\{\tbm a_n[t]\}_{t=P}^{T}$ be generated by JSTIRSO (\textbf{Procedure~2}) with
		\myitem \cmt{constant step size}a constant step size $\alpha \in  (0,1/L]$.
		\myitem \cmt{bounded variations}If there exists $\sigma$ such that
		\begin{equation} \label {eq:boundedvariationsassump1}
		\left \lVert \timevarhindsight[t] - \timevarhindsight[t-1] \right \rVert_2 \leq \sigma,~ \forall\,t\geq P+1,
		\end{equation}
	\end{myitemize}%
	then the dynamic regret of JSTIRSO satisfies:
	\begin{multline}\label{eq:regretboundTIRSO}
	\tilde R_d^{(n)}[T] \leq \frac{1}{\alpha \strongcvxparf}\Big[B_{\bm v} +\lambda \sqrt{N}\Big]  \big(\lVert \tbm a_n[P]- \timevarhindsight[P] \rVert_2 + W^{(n)}[T] \\+ \alpha E^{(n)}[T]\big ),
	\end{multline}
where $B_{\bm v}$ is defined in \eqref{eq:JSTIRSOgradBound}. 
\end{theorem}
\begin{proof}
	\begin{myitemize}%
		\myitem \cmt{Dynamic regret result}In order to derive the dynamic regret of JSTIRSO, since $\tilde h_t$ is convex, we have by definition
		\vspace{-2mm}
		\begin{equation} \label{eq:firstordercvxcond}
		\tilde h_t^{(n)}(\timevarhindsight[t]) \geq \tilde h_t^{(n)}(\tbm a_n[t]) + \left( \nabla^s \tilde h_t^{(n)}(\tbm a_n[t])\right )^\top \!\!\left(  \timevarhindsight[t] \!- \tbm a_n[t]\right ),
		\end{equation}
		$\forall \, \timevarhindsight[t], \tbm a_n[t]$, where a subgradient of $\tilde h_t^{(n)}$ is given by 
		$\nabla^s \tilde h_t^{(n)}(\tbm a_n[t])= \nabla \tilde \rLoss_t^{(n)}(\tbm a_n[t]) + \bm u_t$ with $ \bm u_t \in \partial  \Omega^{(n)}(\tbm a_n[t])$. Rearranging \eqref {eq:firstordercvxcond} and summing both sides of the inequality from $t=P$ to $T$ results in:
		\vspace{-3mm}
		\begin{multline}
		\sum_{t=P}^{T}\left [ \tilde h_t^{(n)}(\tbm a_n[t])- \tilde h_t^{(n)}(\timevarhindsight[t])\right  ] \leq \\
		 \sum_{t=P}^{T} \left(\nabla^s \tilde h_t^{(n)}(\tbm a_n[t])\right )^\top 
		 \left( \tbm a_n[t] -\timevarhindsight[t]\right ).
		\end{multline}
		By applying the Cauchy-Schwarz inequality to each term of the summation in the r.h.s. of the above inequality, we obtain
		\vspace{-3mm}
		\begin{multline} \label{eq:dynamicregretinequality}
		\sum_{t=P}^{T}\left [ \tilde h_t^{(n)}(\tbm a_n[t])- \tilde h_t^{(n)}(\timevarhindsight[t])\right  ]\leq \\
		\sum_{t=P}^{T} \left \lVert \nabla^s \tilde h_t^{(n)}( \tbm a_n[t]) \right \rVert_2 
		\cdot \left\lVert \tbm a_n[t] -\timevarhindsight[t]\right \rVert_2.
		\end{multline}
		The next step is to derive an upper bound on $\lVert \nabla^s \tilde h_t^{(n)}( \tbm a_n[t]) \rVert_2$. From the definition of $ \nabla^s \tilde h_t^{(n)}( \tbm a_n[t]) $ and by the triangular inequality, we have
		\begin{equation} \label{eq:boundonsubgrad}
		\lVert \nabla^s \tilde h_t^{(n)}( \tbm a_n[t]) \rVert_2 \leq \lVert  \nabla  \tilde \rLoss_t^{(n)}( \tbm a_n[t]) \rVert_2 + \left \lVert \bm u_t\right \rVert_2.
		\end{equation}

		From Lemma \ref{lemma:boundedgrad} and Lemma \ref{lemma:boundedsubgrad}, we have $\lVert \nabla^s \tilde h_t^{(n)}( \tbm a_n[t]) \rVert_2 \leq B_{\bm v}+\lambda \sqrt{N}$. Substituting it into \eqref {eq:dynamicregretinequality} leads to:
		\vspace{-2mm}
		\begin{multline} \label{eq:boundTIRSO}
		\sum_{t=P}^{T} \left[ \tilde h_t^{(n)}(\tbm a_n[t])- \tilde h_t^{(n)}(\timevarhindsight[t])\right  ] \\ \leq \sum_{t=P}^{T} \Big [B_{\bm v}+ \lambda \sqrt{N}\Big ]\left \lVert \tbm a_n[t] -\timevarhindsight[t]\right \rVert_2.
		\end{multline}
		\myitem 	
		Next, we apply Lemma 2 
		in \cite{dixit2019onlineproximal} in order to bound $\sum_{t=P}^T\lVert \tbm a_n[t] -\timevarhindsight[t] \rVert_2$ in \eqref {eq:boundTIRSO}. The hypotheses of Lemma 2 are Lipschitz smoothness of $\tilde \rLoss_t^{(n)}$, Lipschitz continuity of $\nReg$, and strong convexity of $\tilde \rLoss_t^{(n)} $. Lipschitz continuity of $\nReg$ is proved in 
		Lemma \ref{lemma:boundedsubgrad} whereas strong convexity of $\tilde \rLoss_t^{(n)} $ is implied by the assumption A\ref{as:mineig}.
		\myitem \cmt{proving Lipschitz smoothness of loss function $f$} 
		To verify that $\rev{ \tilde \rLoss_t^{(n)} }$ is Lipschitz-smooth, it suffices to realize that $\tilde \rLoss_t^{(n)}$ is twice-differentiable, and thus assumption A\ref{as:maxeig} is equivalent to saying that $\tilde \rLoss_t^{(n)}$ is $L$-Lipschitz smooth. 
		\par
		\myitem \cmt{Applying Lemma 2}To apply \cite [Lemma 2]{dixit2019onlineproximal}, one can set the variable $K$ in that context as $T-P+1$, $g_k$ as $\Omega^{(n)}$, and $f_k$ as $\tilde\rLoss_{P+k-1}^{(n)}$, and it follows that $\bm x_k$ in \cite {dixit2019onlineproximal} equals $\tbm a_n[P+k-1]$ and $\bm x_k^\circ$ equals $\tbm a_n^\circ[P+k-1]$. Then, since we have already shown above that the hypotheses of Lemma 2 in \cite {dixit2019onlineproximal} hold in our case, applying it to bound $\lVert \tbm a_n[t] -\timevarhindsight[t]\rVert_2$ in \eqref {eq:boundTIRSO} yields:
		\vspace{-2mm}
		\begin{align*}
		&\sum_{t=P}^{T}\left [ \tilde h_t^{(n)}(\tbm a_n[t])- \tilde h_t^{(n)}(\timevarhindsight[t])\right  ] \leq \nonumber \\
		&\frac{1}{\alpha \strongcvxparf}\Big[ B_{\bm v}  \!+\!\lambda \sqrt{N}\Big]  \left(\lVert \tbm a_n[P]\!- \!\timevarhindsight[P] \rVert_2\! + \!W^{(n)}[T]\!+\! \alpha E^{(n)}[T]\right ).
		\end{align*}
	\end{myitemize}%
	This concludes the proof (note that initializing $\tbm a_n[P]=\bm 0_{NP}$ can lead to further simplification).
\end{proof}

The bound on the dynamic regret for JSTIRSO depends on $W^{(n)}[T]$ and $E^{(n)}[T]$, which formalizes how much the variability and uncertainty affect the parameter estimation. 
This has also been verified experimentally, as it is shown in Section VII, Fig. \ref{fig:dynamicregret}, which shows that for a higher missing probability, the normalized dynamic regret has higher values, as expected. Moreover, when there is an abrupt model transition,  
the normalized regret starts to increase. 
It should be noticed that the theoretical assumptions under which the dynamic regret becomes sublinear in $T$ (sublinear path length $W^{(n)}[T]$, and sublinear cumulative error $E^{(n)}[T]$) may not hold in practice when the model parameters are changing all the time, and in the presence of observation noise and missing data; however, the rate of growth of the regret can be used as a benchmark to compare different approaches.

The cumulative error $E^{(n)}[T]$ can be bounded as a function of the quantities introduced in A\ref{as:errorbounds} (related to the inexactness of the reconstructed samples). The following lemma establishes that under such assumptions, the error on the gradient (i.e.,  $  \lVert \bm e^{(n)}[t] \rVert_2 $) is always bounded.

\begin{mylemma} 
\label{lemma:e}Under assumptions A\ref{as:boundedprocess} and A\ref{as:errorbounds}, let $\{\tbm a_n[t]\}_{t=P}^{T}$ be generated by JSTIRSO (\textbf{Procedure~2}) with
		\myitem \cmt{constant step size}a constant step size $\alpha \in  (0,1/L]$. Then, the error associated with the inexact gradient [cf. \eqref{eq:erroringrad}] is bounded  as $ \lVert \bm e^{(n)}[t] \lVert_2  \leq B_{\bm e}$, where
\begin{multline}
B_{\bm e} \define \left (\gamma B_{\bm\Phi}  + \left(\frac{\nu}{1+\nu}\right)\left(2 \sqrt{PN\energybound} B_{\bm g}+B_{\bm g}^2\right) \right )\\
 \times \frac{\sqrt{PN} \energybound}{\strongcvxparf }\left (\frac{\nu}{1+\nu}+\frac{1}{1-\gamma }\right)
 + \gamma B_{\bm r}+ \left(\frac{\nu}{1+\nu}\right) B_{\bm g} \sqrt{\energybound}.
\end{multline}
\end{mylemma}

\begin{proof}
See Appendix B in the supplementary material.
\end{proof}

This bound depends on three kinds of quantities: a) bounds related to the inexactness of the reconstructed signal, b) simple properties of the data time series, and c) the hyperparameters $\nu$ and $\gamma$. Note that $\lVert \bm e^{(n)}[t] \rVert_2$ and $E^{(n)}[T]$ are related via \eqref{eq:cumulativeerror}. In those cases where the sources of uncertainty are such that $\lVert \bm e^{(n)}[t] \rVert_2$ does not vanish, the above bound can be used to replace $E^{(n)}[T]$ in the regret bound in \eqref{eq:regretboundTIRSO} with an expression that depends on the quantities expressed in A\ref{as:errorbounds}. 
\begin{mycorollary}
Under the hypotheses in Theorem 1, the dynamic regret of JSTIRSO satisfies: 
\vspace{-2mm}
\begin{multline}\label{eq:regretboundTIRSOcorollary}
	\tilde R_d^{(n)}[T] \leq \frac{1}{\alpha \strongcvxparf}\Big[B_{\bm v} +\lambda \sqrt{N}\Big]  \big(\lVert \tbm a_n[P]- \timevarhindsight[P] \rVert_2 + W^{(n)}[T] \\
	+ \alpha T B_{\bm e} \big).
	\end{multline}
\end{mycorollary}
Observe that the above regret bound has a term that is linear in $T$, and this case was commented after Theorem \ref{th:dynamicregretboundTIRSO}. If $W^{(n)}[T]$ is sublinear, then the \textcolor{black}{asymptotic growth rate of the dynamic regret is bounded by $ (B_{\bm v} + \lambda \sqrt{N}) B_{\bm e}/\beta_\ell$. Note that this is a worst-case bound that does not depend on the stepsize $\alpha$.}
\end{myitemize}

Intuitively, dynamic regret characterizes the ability to predict the next signal observation from the estimated parameters and reconstructed signals. A remaining challenge is to determine under which conditions the algorithms are able to identify parameters and signals. This is important because, under identifiability conditions, one could claim that the lower the regret bound is, the closer the reconstructed signals will be to the true signals. Consequently, apart from obtaining a smaller value of $B_{\bm g}$, also $\{\hbm \Phi[t], \hbm r_n[t]\}$ will become closer to the (not observable) $\{\bm \Phi[t], \bm r_n[t]\}$, which will be associated with smaller values of the quantities $B_{\bm\Phi}, B_{\bm r}$. The dependency of these bounds on regret and the interaction between such bounds are topics that lie out of the scope of the present work and could give rise to improved regret bounds.

\section{Experimental Results}\label{sec:simulationsE}
To analyze the performance of the proposed algorithms, we evaluate both the prediction normalized mean squared error (NMSE) for the signal, which is given by:
	\begin{myitemize}%
		\myitem \cmt{NMSD Signal}
		\begin{equation} \label{eq:signalNMSD}
		\text{NMSD}_s[t] =\frac{\mathbb{E}\big[\left  \lVert \bm y[t] - \hbm y[t] \right \rVert_2^2 \big]}{\mathbb{E}\left [\lVert \bm y[t]\lVert_2^2\right ]},
		\end{equation}
		where $\bm y[t]$ is the true signal while $\hbm y[t]$ is the predicted signal; as well as the performance of the topology estimation, which is evaluated by the topology normalized mean squared deviation (NMSD).   
		\myitem \cmt{NMSD}The NMSD for the graph (topology) is defined as:
		\begin{equation} \label{eq:nmsd}
		\text{NMSD}_{g}[t] \define  \frac{\mathbb{E}\big [ \sum_{n=1}^N \lVert  {\bm{a}}_n[t] - \bm{a}_n^\text{true}(t) \lVert_2^2\big ]}{\mathbb{E}\big [ \sum_{n=1}^N\lVert\bm{a}_n^\text{true}(t) \lVert_2^2\big]},
		\end{equation}
		which measures the difference between the estimates $\{\bm a_n[t]\}_{t}$ and the time-varying true VAR coefficients $\{\bm a_n^\text{true}(t)\}_{t}$.  
		\vspace{-2mm}
		\subsection{Synthetic Data}
		\subsubsection{Data generation}
  \label{ss:data_generation}
		\myitem \cmt{Network via Erd\H{o}s-R\'enyi model}We consider a dynamic VAR model where the coefficients change abruptly at two specific points in time. 
		To generate the synthetic data, an
		Erd\H{o}s-R\'enyi random graph is generated with edge probability $p_e$ and
		self-loop probability 1.
		This random graph underlies the data generation and its binary adjacency matrix	
		determines which entries of
		the matrices $\{\bm A^{(t)}_p\}_{p=0}^P$ are zero for all $t$. 
		\myitem \cmt{VAR coefficients from Gaussian distribution}The rest of the entries are drawn i.i.d. from a standard normal distribution.
		Each of the matrices $\{\{\bm
		A^{(t)}_p\}_{p=0}^P\}_{t=1}^T$ is then scaled down by a constant
		that ensures that the VAR process is stable \cite{lutkepohl2005}.
		\myitem \cmt{Time series data generation}The
		innovation process samples are drawn independently as
		$\bm{u}[t]\sim\mathcal{N}(\bm{0}_N, \sigma_u^2\bm
		I_{N})$. At $t=T/3$ and $t=2T/3$, the model changes abruptly from one model to another model, by generating at each changepoint a new set of VAR coefficients with the appropriate support (adjacency matrix). Changes in the adjacency matrix are simulated as follows: at each transition, the adjacency matrix is also changed by altering 33\% of the edges. This means that 1/3 of the edges are removed and new edges are introduced with probability $p_e/3$.
	\end{myitemize}
\setcounter{figure}{1}
	\begin{figure}
		\centering
		\includegraphics[width=\linewidth]{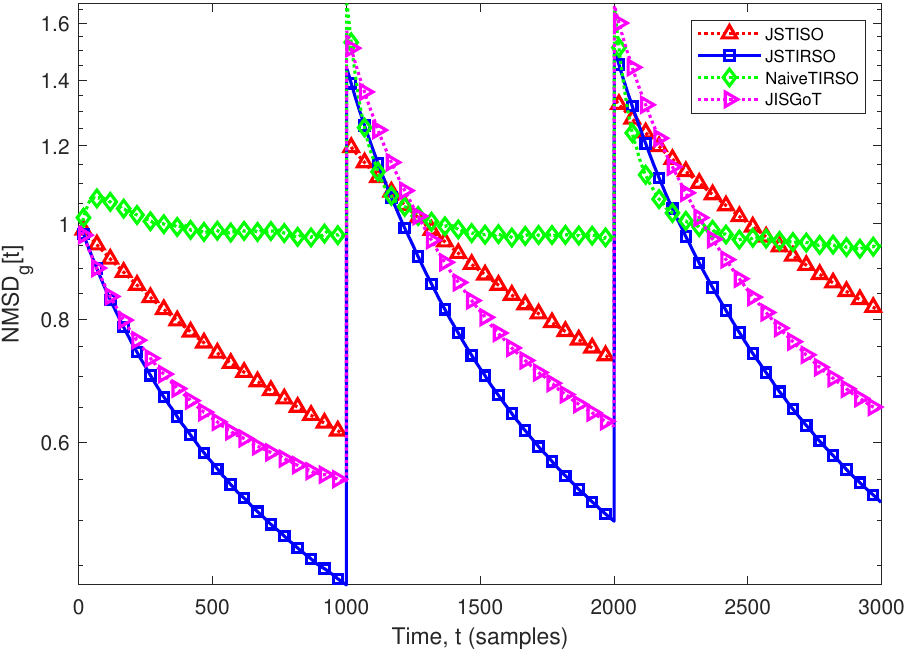}
		\caption{NMSD vs. time,  Simulation parameters: $N=10, P=4, T=3000, \sigma_u=0.01, \sigma_\epsilon=0.01, \gamma=0.8, \rho=0.75, p_e=0.1, \alpha=\zeta/L, \zeta \in (0,1]$, number of Monte Carlo (MC) iterations = 300, JISGoT iterations = 20.}
		\label{fig:nmsd-topology}
	\end{figure}
	\begin{figure}
		\centering
		\includegraphics[width=\linewidth]{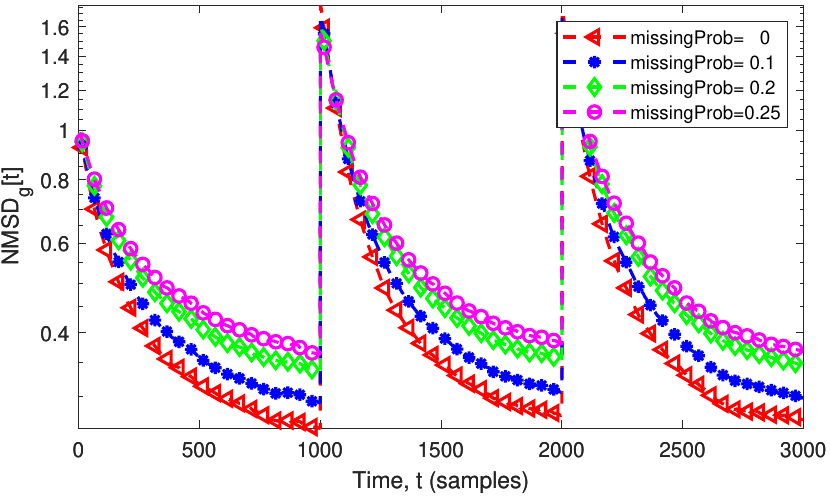}
		\caption{NMSD of JSTIRSO for topology estimation for various missing prob. vs. time. Simulation parameters: $N=8, P=3, T=3000, p_e=0.2, \sigma_u=0.01, \sigma_\epsilon=0.01, \gamma=0.9, \alpha=\zeta/L, \zeta \in (0,1]$, number of MC iterations = 100.}
		\label{fig:nmsd-topology-multiple}
	\end{figure}
		\begin{figure}
		\centering
		\includegraphics[width=\linewidth]{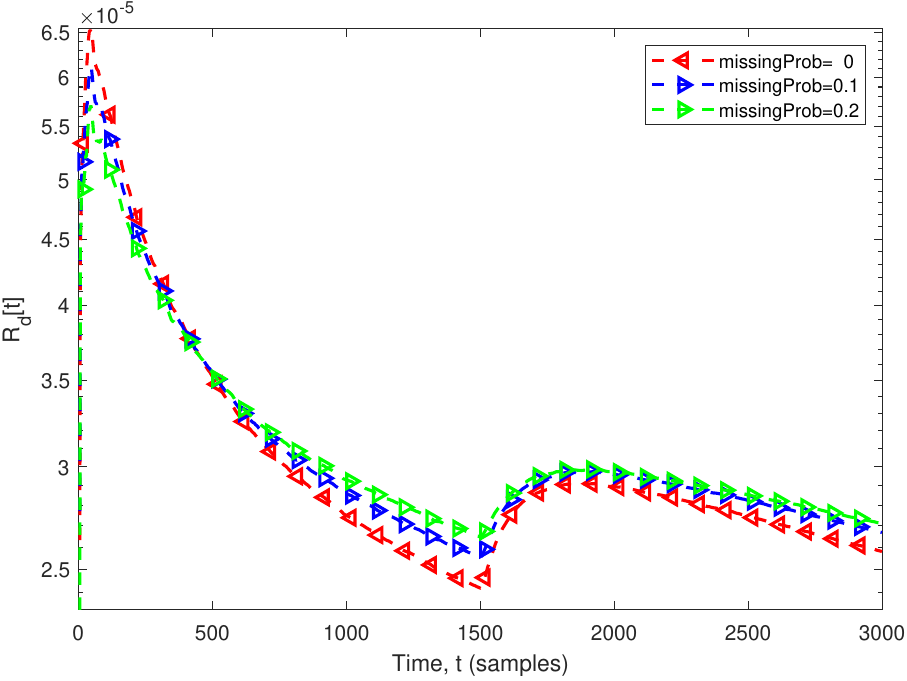}
		\caption{Cumulative dynamic regret (normalized by the number of instants for which there are no missing values), given by $R_d[T]= \sum_{n=1}^N \tilde R_d^{(n)}[T]/ \sum_{n=1}^N \sum_{t=P}^T \mathds{1} \{m_n[t]=1\}$ of JSTIRSO for topology estimation for various missing probabilities vs. time. Simulation parameters: $N=8, P=3, T=3000, p_e=0.2, \sigma_u=0.01, \sigma_\epsilon=0, \gamma=0.9, \alpha=\zeta/L, \zeta \in (0,1]$, number of MC iterations = 30.}
		\label{fig:dynamicregret}
	\end{figure}
	\subsubsection{Competing Alternatives}
	The performance of JSTISO (Procedure \ref{alg:JSTISO}) and JSTIRSO (Procedure \ref{alg:JSTIRSO}) is evaluated and compared with that corresponding to two competing alternatives. 
    The first alternative to our algorithm is a simple procedure based on TIRSO \cite{zaman2019online}, where the missing values are imputed directly as their predicted values via the VAR model \eqref{eq:model}, and the noisy samples are not refined: this procedure is referred to as `NaiveTIRSO'. The second alternative is an adaptation of the  JISGoT algorithm~\cite[Algorithm 4]{ioannidis2019semiblindinference}, which is, to the best of our knowledge, the state-of-art in joint signal and topology estimation. The JISGoT algorithm refines the previous $P$ signal estimates and runs several iterations at each time instant, incurring a computational complexity of $\mathcal{O}(KN^3P^2)$ per time instant, where $K$ is the number of iterations used in the inner loop that refines the signal estimates.
	The values for the parameters $\alpha,~ \nu, ~\gamma,$ and $\lambda$  in JSTISO, JSTIRSO, and JISGoT are selected via grid search to minimize the squared deviation for a validation signal.
	\subsubsection{Discussion of results}In Fig. \ref{fig:nmsd-topology}, the NMSD for the topology estimation [cf. \eqref{eq:nmsd}] is presented for the four algorithms described above. The input data are generated as described in Sec.~\ref{ss:data_generation} where the underlying adjacency matrix changes at $t=T/3$ and $t=2T/3$. 
	The NMSD obtained by NaiveTIRSO saturates near 1. JSTISO tracks the topology more slowly than JSTIRSO or JISGoT, as expected since JSTISO disregards the past completely. JSTIRSO achieves a lower NMSD eventually as compared to JISGoT. Notice also that JSTIRSO requires less computation than JISGoT. As expected, due to the careful choice of the loss function in JSTIRSO, it attains a lower level of NMSD$_g$ than that of JISGoT, despite JSTIRSO does not refine the previous signal estimates.
	
	Fig. \ref{fig:nmsd-topology-multiple} presents a comparison of JSTIRSO applied to different data sets generated using various missing probabilities in the observations. The performance of JSTIRSO for multiple values of missing probability is compared with that of zero missing probability. As expected, the figure shows that the higher the missing probability in the observations is, the higher the NMSD for the graph estimation is. In Fig. \ref{fig:dynamicregret}, for the sake of simplicity in the illustration, a transition point at $T/2$ is considered, and the cumulative normalized dynamic regret (normalized by the number of instants when the value is not missing), given by $R_d[T]= \sum_{n=1}^N \tilde R_d^{(n)}[T]/ \sum_{n=1}^N \sum_{t=P}^T \mathds{1} \{m_n[t]=1\}$, is presented for JSTIRSO for different values of missing probabilities. As expected, the result in Fig. \ref{fig:dynamicregret} shows that when the missing probability value is decreased, the cumulative normalized dynamic regret has also lower values. The same trend is followed after the transition point at $T/2$.
	\vspace{-2mm}
\subsection{Real Data}
In this section, we present the results obtained using real data. The real data are taken from Lundin’s offshore oil and gas (O\&G) platform Edvard-Grieg.\footnote{https://www.lundin-petroleum.com/operations/production/norway-edvard-grieg} We use a dataset containing 24-time
series corresponding to the main 24 variables of the decantation system that separates oil, gas, and water. Each node corresponds to a temperature, pressure, or oil-level sensor placed in the aforementioned subsystem. Causal relations among these time series are expected since they are physically coupled due to the pipelines connecting the various system parts, and due to the inherent control systems therein. Topology identification is motivated to predict the short-term future values of the time series corresponding to temperature, pressure, and oil-level sensors and to unveil dependencies that cannot be inferred by simple human inspection. All time series are re-sampled having a common sampling period using linear interpolation.
Each time series is also normalized to have zero mean and unit sample standard deviation.

The results in Fig. \ref{fig:real-data-NMSE} represent the performance of JSTIRSO for different missing probabilities of observation in the real data, by presenting the prediction NMSE for each case. First, the hyperparameters of JSTIRSO are cross-validated via grid search. Then, using the cross-validated hyperparameters, the prediction NMSE versus time is plotted corresponding to each value of missing probability in the observations. The missing values are synthetically introduced to the real data using the model in \eqref{eq:observationmodel}. The results show that when the missing probability increases, the prediction NMSE of JSTIRSO also increases, as expected. 

Fig. 6 (in the supplementary material) 
displays the average graphs estimated via JSTIRSO by thresholding the average of the estimated VAR coefficients across the intervals $[k/(3T), (k+1)/(3T)], k=0,1,3$ for missing probabilities $0$, $0.05$, $0.15$ and $0.2$. One can observe that the average estimated graph changes with time since the underlying system is dynamic. Moreover, the results illustrate that the proposed algorithm JSTIRSO is robust and can estimate the graph when there are some missing values in the observations. 
	
	\begin{figure}
		\centering
		\includegraphics[width=\linewidth]{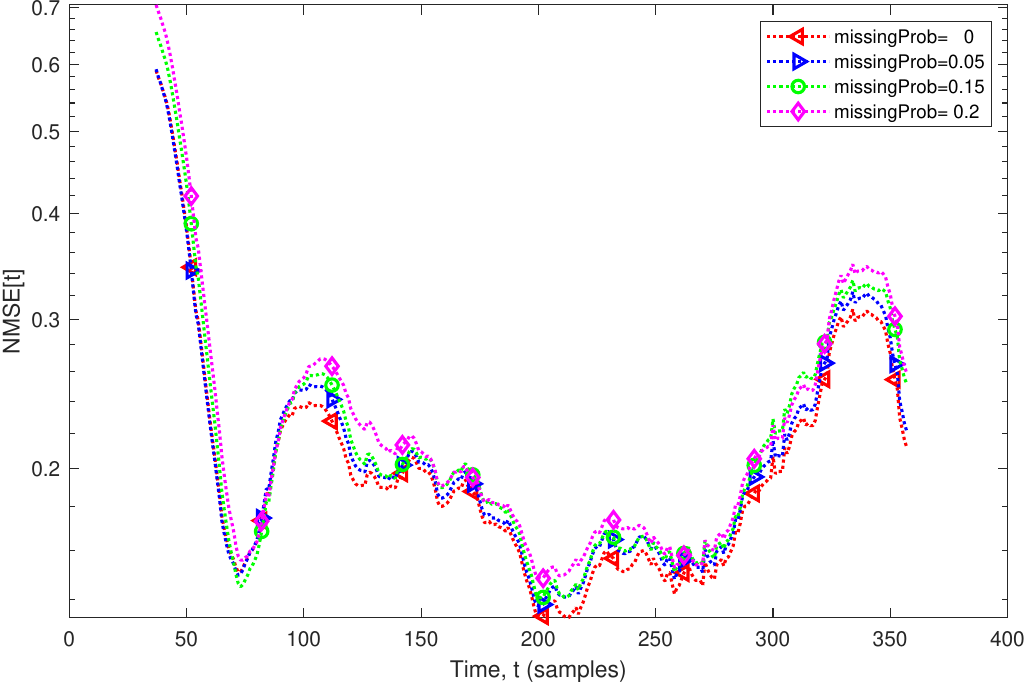}
		\caption{Prediction NMSE for JSTIRSO vs. time using real data from Lundin. Simulation parameters: $N=24, P=6, T=360, \sigma_u=0.01, \sigma_\epsilon=0.01, \gamma=0.9$.}
		\label{fig:real-data-NMSE}
	\end{figure}
\nextver{	
	\begin{figure*}
		\centering
		\includegraphics[width=\textwidth,height=24cm]{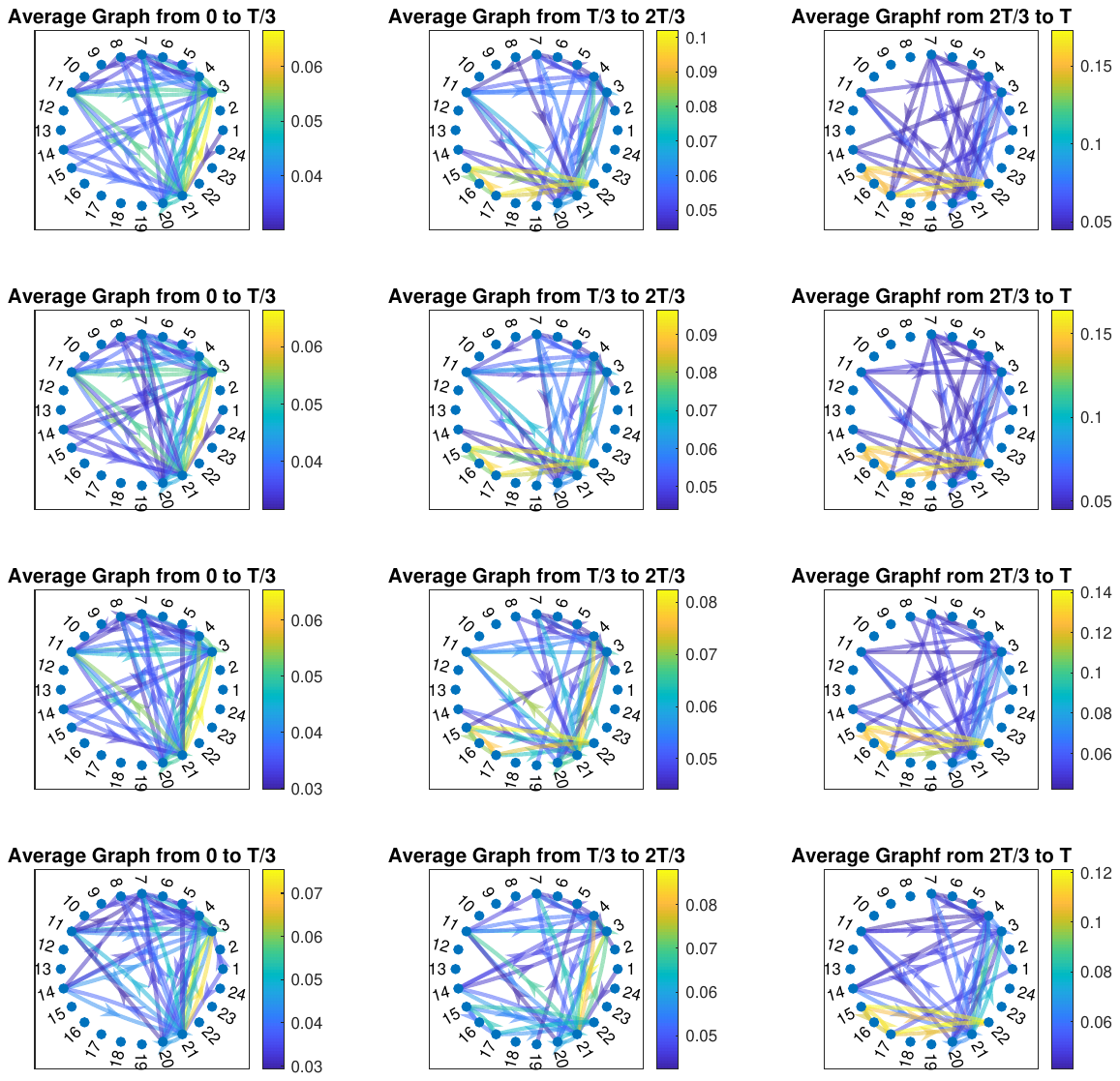}
		\caption{Average estimated graphs via JSTIRSO. The rows of the sub-figures correspond to missing probabilities 0, 0.05, 0.15, and 0.2. respectively. All the hyperparameters of the algorithm are computed via grid search.  Simulation parameters: $N=24, P=6, T=360, \sigma_u=0.01, \sigma_\epsilon=0.01, \gamma=0.9$.}
		\label{fig:real-data-avgGraphs}
	\end{figure*}
	}
	\section{Conclusions}
	\label{sec:conclusions}
	To track time-varying topologies from noisy observations in the presence of missing data, the online algorithm JSTIRSO has been proposed by minimizing an online joint optimization criterion. 
	Thanks to a carefully formulated loss function, joint signal and topology estimation can be carried out efficiently (especially in the case of the low-complexity JSTISO); moreover, the performance of JSTIRSO has been characterized theoretically. To this end, a dynamic regret bound has been derived as a function of the path length (which quantifies the variation in the topologies) and cumulative error on the gradient (which quantifies the effect of noise and missing values). The error on the gradient is in turn bounded [cf. Lemma \ref{lemma:e}] as a function of the maximum deviation of the (inexact) estimate of the auxiliary variables 
    from their associated true values. The bound on the dynamic regret becomes sublinear in scenarios where the variation in the time-varying topologies, the probability of missing data, and the observation noise level are vanishing with time. Numerical results have shown that JSTIRSO can track the time-varying topologies from noisy observations with missing values with a smaller deviation than (state-of-the-art) JISGoT. Future research avenues include the combination of the proposed strategy with tools related to Kalman filtering and smoothing \cite{ioannidis2019semiblindinference} to ascertain its improvement in terms of performance.

	\if\editmode1 
		\onecolumn
		\printbibliography
	\else
		\bibliography{\bibfilenames}

\begin{thebibliography}{10}

\bibitem{kolaczyck2009}
E.~D. Kolaczyk,
\newblock {\em Statistical Analysis of Network Data: Methods and Models},
\newblock Springer, New York, 2009.

\bibitem{isufi2018forecasting}
E.~Isufi, A.~Loukas, N.~Perraudin, and G.~Leus,
\newblock ``Forecasting time series with varma recursions on graphs,''
\newblock {\em IEEE Trans. Signal Process.}, vol. 67, no. 18, pp. 4870--4885, 2019.

\bibitem{lorenzo2016lms}
P.~D. Lorenzo, S.~Barbarossa, P.~Banelli, and S.~Sardellitti,
\newblock ``Adaptive least mean squares estimation of graph signals,''
\newblock {\em IEEE Trans. Signal Info. Process. Netw.}, vol. 2, no. 4, pp. 555--568, Dec. 2016.

\bibitem{liu2016unsupervised}
C.~Liu, S.~Ghosal, Z.~Jiang, and S.~Sarkar,
\newblock ``An unsupervised spatiotemporal graphical modeling approach to anomaly detection in distributed {CPS},''
\newblock in {\em ACM/IEEE Int. Conf. Cyber-Physical Syst.}, Apr. 2016, pp. 1--10.

\bibitem{shen2017dimensionalityreduction}
Y.~Shen, P.~A. Traganitis, and G.~B. Giannakis,
\newblock ``Nonlinear dimensionality reduction on graphs,''
\newblock in {\em Proc. IEEE Int. Workshop Comput. Advan. Multi-Sensor Adapt. Process.}, Curacao, Netherlands Antilles, Dec. 2017.

\bibitem{hoeltgebaum2021estimation}
H.~Hoeltgebaum, N.~Adams, and C.~Fernandes,
\newblock ``Estimation, forecasting, and anomaly detection for nonstationary streams using adaptive estimation,''
\newblock {\em IEEE Trans. Cybern.}, vol. 52, no. 8, pp. 7956--7967, 2021.

\bibitem{little2014}
R.~J.~A. Little and D.~B. Rubin,
\newblock {\em Statistical Analysis with Missing Data},
\newblock John Wiley \& Sons, Inc., USA, 2014.

\bibitem{pavez2019missingdata}
E.~Pavez and A.~Ortega,
\newblock ``Covariance matrix estimation with non uniform and data dependent missing observations,''
\newblock {\em IEEE Trans. Inf. Theory}, vol. 67, no. 2, pp. 1201--1215, 2020.

\bibitem{humbert2009better}
J.~Y. Humbert, L.~S. Mills, J.~S. Horne, and B.~Dennis,
\newblock ``A better way to estimate population trends,''
\newblock {\em Oikos}, vol. 118, no. 12, pp. 1940--1946, 2009.

\bibitem{clark2004population}
J.~S. Clark and O.~N. Bj{\o}rnstad,
\newblock ``Population time series: process variability, observation errors, missing values, lags, and hidden states,''
\newblock {\em Ecology}, vol. 85, no. 11, pp. 3140--3150, 2004.

\bibitem{harvey1984estimating}
A.~C. Harvey and R.~G. Pierse,
\newblock ``Estimating missing observations in economic time series,''
\newblock {\em Journal of the American Statistical Association}, vol. 79, no. 385, pp. 125--131, 1984.

\bibitem{little2019statistical}
R.~J.~A. Little and D.~B. Rubin,
\newblock {\em Statistical Analysis with Missing Data}, vol. 793,
\newblock John Wiley \& Sons, 2019.

\bibitem{grover2021MissingDataSpatioTemporal}
A.~Grover and B.~Lall,
\newblock ``A recursive method for estimating missing data in spatio-temporal applications,''
\newblock {\em IEEE Trans. Ind. Inform.}, vol. 18, no. 4, pp. 2714--2723, 2021.

\bibitem{adhikari2022missingdataIoT}
D.~Adhikari, W.~Jiang, J.~Zhan, D.~B. Rawat, U.~Aickelin, and H.~A. Khorshidi,
\newblock ``A comprehensive survey on imputation of missing data in internet of things,''
\newblock {\em ACM Comput. Surveys}, vol. 55, no. 7, pp. 1--38, 2022.

\bibitem{kim2022missingvaluesship}
Y.~Kim, S.~Steen, and H.~Muri,
\newblock ``A novel method for estimating missing values in ship principal data,''
\newblock {\em Ocean Eng.}, vol. 251, pp. 110979, 2022.

\bibitem{zhang2022missingenvironmental}
Y.~Zhang and P.~J. Thorburn,
\newblock ``Handling missing data in near real-time environmental monitoring: A system and a review of selected methods,''
\newblock {\em Future Generation Computer Systems}, vol. 128, pp. 63--72, 2022.

\bibitem{pan2022imputation}
Z.~Pan, Y.~Wang, K.~Wang, H.~Chen, C.~Yang, and W.~Gui,
\newblock ``Imputation of missing values in time series using an adaptive-learned median-filled deep autoencoder,''
\newblock {\em IEEE Trans. Cybern.}, vol. 53, no. 2, pp. 695--706, 2022.

\bibitem{karimi2021joint}
H.~S. Karimi and B.~Natarajan,
\newblock ``Joint topology identification and state estimation in unobservable distribution grids,''
\newblock {\em IEEE Trans. on Smart Grid}, vol. 12, no. 6, pp. 5299--5309, 2021.

\bibitem{mateos2018connecting}
G.~Mateos, S.~Segarra, A.~G. Marques, and A.~Ribeiro,
\newblock ``Connecting the dots: Identifying network structure via graph signal processing,''
\newblock {\em IEEE Signal Process. Mag.}, vol. 36, no. 3, pp. 16--43, 2019.

\bibitem{angelosante2011graphical}
D.~Angelosante and G.~B. Giannakis,
\newblock ``Sparse graphical modeling of piecewise-stationary time series,''
\newblock in {\em Proc. IEEE Int. Conf. Acoust., Speech, Signal Process.}, Prague, Czech Republic, 2011, pp. 1960--1963.

\bibitem{segarra2017templates}
S.~Segarra, A.~G. Marques, G.~Mateos, and A.~Ribeiro,
\newblock ``Network topology inference from spectral templates,''
\newblock {\em IEEE Trans. Signal Info. Process. Netw.}, vol. 3, no. 3, pp. 467--483, Sep. 2017.

\bibitem{kline2015}
R.~B. Kline,
\newblock {\em Principles and Practice of Structural Equation Modeling},
\newblock Guilford Publications, 2015.

\bibitem{shen2017tensor}
Y.~Shen, B.~Baingana, and G.~B. Giannakis,
\newblock ``Tensor decompositions for identifying directed graph topologies and tracking dynamic networks,''
\newblock {\em IEEE Trans. Signal Process.}, vol. 65, no. 14, pp. 3675--3687, Jul. 2017.

\bibitem{bishop2006}
C.~M. Bishop,
\newblock {\em Pattern Recognition and Machine Learning},
\newblock Information Science and Statistics. Springer, 2006.

\bibitem{granger1988causality}
C.~W.~J. Granger,
\newblock ``Some recent development in a concept of causality,''
\newblock {\em J. Econometrics}, vol. 39, no. 1-2, pp. 199--211, Sep. 1988.

\bibitem{zellner1979causality}
A.~Zellner,
\newblock ``Causality and econometrics,''
\newblock in {\em Carnegie-Rochester Conference series on Public Policy}. Elsevier, 1979, vol.~10, pp. 9--54.

\bibitem{kay1}
S.~M. Kay,
\newblock {\em Fundamentals of Statistical Signal Processing, {V}ol. {I}: Estimation Theory},
\newblock Prentice-Hall, 1993.

\bibitem{goebel2003varcausality}
R.~Goebel, A.~Roebroeck, D.S. Kim, and E.~Formisano,
\newblock ``Investigating directed cortical interactions in time-resolved f{MRI} data using vector autoregressive modeling and {G}ranger causality mapping,''
\newblock {\em Magnet. Reson. Imag.}, vol. 21, no. 10, pp. 1251--1261, 2003.

\bibitem{basu2015granger}
S.~Basu, A.~Shojaie, and G.~Michailidis,
\newblock ``Network {G}ranger causality with inherent grouping structure.,''
\newblock {\em J. Mach. Learn. Res.}, vol. 16, no. 2, pp. 417--453, Mar. 2015.

\bibitem{bach2004learning}
F.~R. Bach and M.~I. Jordan,
\newblock ``Learning graphical models for stationary time series,''
\newblock {\em IEEE Trans. Signal Process.}, vol. 52, no. 8, pp. 2189--2199, Aug. 2004.

\bibitem{songsiri2010selection}
J.~Songsiri and L.~Vandenberghe,
\newblock ``Topology selection in graphical models of autoregressive processes,''
\newblock {\em J. Mach. Learn. Res.}, vol. 11, pp. 2671--2705, Oct. 2010.

\bibitem{bolstad2011groupsparse}
A.~Bolstad, B.~D.~Van Veen, and R.~Nowak,
\newblock ``Causal network inference via group sparse regularization,''
\newblock {\em IEEE Trans. Signal Process.}, vol. 59, no. 6, pp. 2628--2641, Jun. 2011.

\bibitem{songsiri2013vargranger}
J.~Songsiri,
\newblock ``Sparse autoregressive model estimation for learning {G}ranger causality in time series,''
\newblock in {\em Proc. IEEE Int. Conf. Acoust., Speech, Signal Process.}, Vancouver, BC, May 2013, pp. 3198--3202.

\bibitem{mei2017causal}
J.~Mei and J.~M.~F. Moura,
\newblock ``Signal processing on graphs: Causal modeling of unstructured data,''
\newblock {\em IEEE Trans. Signal Process.}, vol. 65, no. 8, pp. 2077--2092, Apr. 2017.

\bibitem{kolar2010estimating}
M.~Kolar, L.~Song, A.~Ahmed, and E.~P. Xing,
\newblock ``Estimating time-varying networks,''
\newblock {\em Ann. Appl. Statist}, pp. 94--123, 2010.

\bibitem{yamada2020timevarying}
K.~Yamada, Y.~Tanaka, and A.~Ortega,
\newblock ``Time-varying graph learning with constraints on graph temporal variation,''
\newblock {\em arXiv preprint arXiv:2001.03346}, 2020.

\bibitem{lopezramos2018dynamic}
L.~M. Lopez-Ramos, D.~Romero, B.~Zaman, and B.~Beferull-Lozano,
\newblock ``Dynamic network identification from non-stationary vector auto-regressive time series,''
\newblock in {\em Proc. IEEE Global Conf. Signal Inf. Process.}, Anaheim, CA, Nov. 2018, pp. 773--777.

\bibitem{yuan2021joint}
Y.~Yuan, D.~W. Soh, X.~Yang, K.~Guo, and T.~Q.~S. Quek,
\newblock ``Joint network topology inference via structured fusion regularization,''
\newblock {\em IEEE Trans. Knowl. Data Eng.}, 2023.

\bibitem{hallac2017network}
D.~Hallac, Y.~Park, S.~Boyd, and J.~Leskovec,
\newblock ``Network inference via the time-varying graphical lasso,''
\newblock in {\em Proc. ACM SIGKDD Int. Conf. Knowl. Discov. Data Min.}, 2017, pp. 205--213.

\bibitem{baingana2014trackingcascades}
B.~Baingana, G.~Mateos, and G.~B. Giannakis,
\newblock ``Proximal-gradient algorithms for tracking cascades over social networks,''
\newblock {\em IEEE J. Sel. Topics Signal Process.}, vol. 8, no. 4, pp. 563--575, Aug. 2014.

\bibitem{zaman2020dynamic}
B.~Zaman, L.~M. Lopez-Ramos, and B.~Beferull-Lozano,
\newblock ``Dynamic regret analysis for online tracking of time-varying structural equation model topologies,''
\newblock in {\em Proc. IEEE Conf. Ind. Electron. Appl. (ICIEA)}, 2020, pp. 939--944.

\bibitem{shafipour2019onlinetopology}
R.~{Shafipour}, A.~{Hashemi}, G.~{Mateos}, and H.~{Vikalo},
\newblock ``Online topology inference from streaming stationary graph signals,''
\newblock in {\em IEEE Data Sci. Workshop}, Jun. 2019, pp. 140--144.

\bibitem{zhang2021online}
X.~Zhang,
\newblock ``Online graph learning in dynamic environments,''
\newblock in {\em Proc. European Signal Process. Conf.}, Belgrade, Serbia, Oct. 2022, pp. 2151--2155.

\bibitem{jiang2021online}
Y.~Jiang, J.~Bigot, and S.~Maabout,
\newblock ``Online graph topology learning from matrix-valued time series,''
\newblock {\em arXiv preprint arXiv:2107.08020}, 2021.

\bibitem{shen2018online}
Y.~Shen and G.~B. Giannakis,
\newblock ``Online identification of directional graph topologies capturing dynamic and nonlinear dependencies,''
\newblock in {\em IEEE Data Sci. Workshop}, 2018, pp. 195--199.

\bibitem{liu2019smoothgraphlearning}
Y.~{Liu}, L.~{Yang}, G.~{Wenbin}, T.~{Peng}, and W.~{Wang},
\newblock ``Spatiotemporal smoothness-based graph learning method for sensor networks,''
\newblock in {\em Proc. IEEE Wireless Commun. Network. Conf.}, Marrakesh, Morocco, 2019, pp. 1--6.

\bibitem{berger2020efficient}
P.~Berger, G.~Hannak, and G.~Matz,
\newblock ``Efficient graph learning from noisy and incomplete data,''
\newblock {\em IEEE Trans. Signal Info. Process. Netw.}, vol. 6, pp. 105--119, 2020.

\bibitem{rao2017estimation}
M.~Rao, T.~Javidi, Y.~C. Eldar, and A.~Goldsmith,
\newblock ``Estimation in autoregressive processes with partial observations,''
\newblock in {\em Proc. IEEE Int. Conf. Acoust., Speech, Signal Process.}, New Orleans, LA, Jun. 2017, pp. 4212--4216.

\bibitem{loh2012missingdatavar}
P.~L. Loh and M.~J. Wainwright,
\newblock ``High-dimensional regression with noisy and missing data: Provable guarantees with nonconvexity,''
\newblock {\em The Annals of Statistics}, pp. 1637--1664, 2012.

\bibitem{coutino2021state}
M.~Coutino, E.~Isufi, T.~Maehara, and G.~Leus,
\newblock ``State-space based network topology identification,''
\newblock in {\em Proc. European Signal Process. Conf.}, Amsterdam, Netherlands, Dec. 2021, pp. 1055--1059.

\bibitem{jiang2020recovery}
J.~{Jiang}, D.~{Tay}, Q.~{Sun}, and S.~{Ouyang},
\newblock ``Recovery of time-varying graph signals via distributed algorithms on regularized problems,''
\newblock {\em IEEE Trans. Signal Info. Process. Netw.}, 2020.

\bibitem{anava2015online}
O.~Anava, E.~Hazan, and A.~Zeevi,
\newblock ``Online time series prediction with missing data,''
\newblock in {\em Proc. Int. Conf. Mach. Learn.}, Lille, France, 2015, pp. 2191--2199.

\bibitem{yang2019online}
H.~Yang and Q.~Pan, Z.and~Tao,
\newblock ``Online learning for time series prediction of ar model with missing data,''
\newblock {\em Neural Process. Lett.}, vol. 50, no. 3, pp. 2247--2263, 2019.

\bibitem{ioannidis2019semiblindinference}
V.~N. {Ioannidis}, Y.~{Shen}, and G.~B. {Giannakis},
\newblock ``Semi-blind inference of topologies and dynamical processes over dynamic graphs,''
\newblock {\em IEEE Trans. Signal Process.}, vol. 67, no. 9, pp. 2263--2274, May 2019.

\bibitem{zaman2019online}
B.~Zaman, L.~M.~Lopez Ramos, D.~Romero, and B.~Beferull-Lozano,
\newblock ``Online topology identification from vector autoregressive time series,''
\newblock {\em IEEE Trans. Signal Process.}, vol. 69, pp. 210--225, 2021.

\bibitem{lutkepohl2005}
H.~Lütkepohl,
\newblock {\em New Introduction to Multiple Time Series Analysis},
\newblock Springer, 2005.

\bibitem{kilian2017}
L.~Kilian and H.~Lütkepohl,
\newblock {\em Structural Vector Autoregressive Analysis},
\newblock Cambridge University Press, 2017.

\bibitem{geiger2015causalinferenceVAR}
P.~Geiger, K.~Zhang, B.~Schoelkopf, M.~Gong, and D.~Janzing,
\newblock ``Causal inference by identification of vector autoregressive processes with hidden components,''
\newblock in {\em Proc. Int. Conf. Mach. Learn.}, 2015, pp. 1917--1925.

\bibitem{dixit2019onlineproximal}
R.~{Dixit}, A.~S. {Bedi}, R.~{Tripathi}, and K.~{Rajawat},
\newblock ``Online learning with inexact proximal online gradient descent algorithms,''
\newblock {\em IEEE Trans. Signal Process.}, vol. 67, no. 5, pp. 1338--1352, Mar. 2019.

\bibitem{shalev2011online}
S.~Shalev-Shwartz,
\newblock ``Online learning and online convex optimization,''
\newblock {\em Found. Trends Mach. Learn.}, vol. 4, no. 2, pp. 107--194, 2011.

\bibitem{zinkevich2003online}
M.~Zinkevich,
\newblock ``Online convex programming and generalized infinitesimal gradient ascent,''
\newblock in {\em Proc. Int. Conf. Mach. Learn.}, 2003, pp. 928--936.

\bibitem{parikh2014proximal}
N.~Parikh and S.~Boyd,
\newblock ``Proximal algorithms,''
\newblock {\em Found. Trends Optim.}, vol. 1, no. 3, pp. 127--239, 2014.

\bibitem{beck2017}
A.~Beck,
\newblock {\em First-Order Methods in Optimization},
\newblock Society for Industrial and Applied Mathematics, 2017.

\end{thebibliography}
	\fi

 	\appendices

	\section{Proof of  Lemma 1}	\label{sec:proof:lemmaboundedgrad}
	\setcounter{equation}{55}
To bound $\lVert  \nabla  \tilde \rLoss_t^{(n)}( \tbm a_n[t]) \rVert_2$, taking the norm on both sides of \eqref{eq:gradfJSTIRSO} and applying  the triangular inequality yields
	\begin{align}
		&\left \lVert \nabla \tilde \rLoss_t^{(n)} (\tbm a_n[t]) \right \rVert_2 \nonumber \\&\leq \left \lVert U_n[t] \hbm g[t] \hbm g^\top [t] \bm a_n[t]\right \rVert_2 +\left \lVert\tilde U_n[t] \tilde y_n[t]  \hbm g[t]\right \rVert_2 \nonumber \\
		& \quad + \left \lVert \gamma \hbm \Phi[t-1] \bm a_n[t] \right \rVert_2 +\left \lVert \gamma \hbm r_n[t-1]\right \rVert_2 \nonumber \\ 
		& \leq U_n[t] \lambda_{\mathrm{max}}\left (\hbm g[t] \hbm g^\top [t]\right ) \left \lVert   \bm a_n[t]\right \rVert_2 +U_n[t]\left \lVert\tilde y_n[t]  \hbm g[t]\right \rVert_2
		\nonumber \\
		& \quad
		+  \gamma \lambda_{\mathrm{max}}\left(\hbm \Phi[t-1]\right ) \left \lVert \bm a_n[t] \right \rVert_2 +\left \lVert \gamma \hbm r_n[t-1]\right \rVert_2.
		\label{eq:generalgradbound} 
	\end{align}
Next, using assumptions A\ref{as:boundedprocess} and A\ref{as:errorbounds}, it can be easily shown that $\lambda_{\mathrm{max}}(\hbm g[t] \hbm g^\top [t]) \leq PN \energybound+ 2 \sqrt{PN \energybound} B_{\bm g} + B_{\bm g}^2$. Substituting this bound in the above expression and using assumption A\ref{as:maxeig} yields
\begin{align}\label{eq:boundGradJSTIRSO}
	&\left \lVert \nabla \tilde \rLoss_t^{(n)} (\tbm a_n[t]) \right \rVert_2 \nonumber \\
	& \leq U_n[t] \left(PN \energybound+ 2 \sqrt{PN \energybound} B_{\bm g} + B_{\bm g}^2\right) \left \lVert   \bm a_n[t]\right \rVert_2 
	\nonumber \\
	& \quad +U_n[t]\sqrt{PN} \energybound
	+  \gamma L \left \lVert \bm a_n[t] \right \rVert_2 +\left \lVert \gamma \hbm r_n[t-1]\right \rVert_2.
\end{align}
Next, an upper bound of $\hbm r_n[t-1]$ is derived. By the definition of $\hbm r_n[t]$ and assumption A\ref{as:boundedprocess}, we have 
\begin{subequations}
	\begin{align}
		\left\lVert \hbm r_n[t-1]  \right\rVert_2 
		& = \left \lVert   \sum_{\tau=P}^{t-1} \,  \gamma^{t-1-\tau} \hat y_n[\tau] \,  \hbm g[\tau] \right \rVert_2 \nonumber  \\
		& \leq  \frac{1}{\gamma}\left \lVert   \sum_{\tau=P}^{t-1} \,  \gamma^{t-\tau} \sqrt{\energybound} \sqrt{\energybound} \bm 1_{NP} \right \rVert_2   \\
		& =   \frac{1}{\gamma}\energybound \sqrt{PN} \gamma^{t} \sum_{\tau=P}^{t-1} \left (\frac{1}{\gamma}\right )^\tau \nonumber  \\
		&=  \frac{1}{\gamma}\energybound \sqrt{PN} \frac{\gamma(1-\gamma^{t-P})}{1-\gamma} \leq  \frac{\sqrt{PN} \energybound}{1-\gamma }. \label{eq:boundonr}	
	\end{align}
\end{subequations}
Using the above bound in \eqref{eq:boundGradJSTIRSO}, it follows that
\begin{subequations} \label{eq:boundGradJSTIRSO2}
	\begin{align}
		&\left \lVert \nabla \tilde \rLoss_t^{(n)} (\tbm a_n[t]) \right \rVert_2 \nonumber \\
		& \leq U_n[t] \left(PN \energybound+ 2 \sqrt{PN \energybound} B_{\bm g} + B_{\bm g}^2\right) \left \lVert   \bm a_n[t]\right \rVert_2 \nonumber \\
		& \quad +U_n[t]\sqrt{PN} \energybound
		+  \gamma L \left \lVert \bm a_n[t] \right \rVert_2 +  \frac{\gamma\sqrt{PN} \energybound}{1-\gamma }\\
		& \leq\frac{\nu}{1+\nu} \left(PN \energybound\!+\! 2 \sqrt{PN \energybound} B_{\bm g} \!+\! B_{\bm g}^2 \!+\! \gamma L\frac{1+\nu}{\nu}\right )\! \left \lVert   \bm a_n[t]\right \rVert_2 
		\nonumber \\
		& \quad +\frac{\nu}{1+\nu} \sqrt{PN} \energybound
		+   \frac{\gamma\sqrt{PN} \energybound}{1-\gamma }. \label{eq:gradboundJSTIRSO}
	\end{align}
\end{subequations} 
The next step is to derive a bound on $ \lVert \bm a_n[t] \rVert_2$. To this end, from \eqref{eq:updateproxonline} and \eqref{eq:gradfJSTIRSO}, it follows that
\begin{align}
	&\left \lVert \bm a_n[t+1] \right  \rVert_2 \nonumber \\
	& \leq \left \lVert \bm a_n^{\text{f}}[t] \right  \rVert_2 \nonumber \\
	& = \left \lVert \bm a_n[t] - \alpha_t \hbm v_n[t]\right  \rVert_2 \nonumber \\
	&=\Big \lVert \bm a_n[t] - \alpha_t \Big ( U_n[t] \hbm g[t] \hbm g^\top [t] \bm a_n[t] - U_n[t] \tilde y_n[t]  \hbm g[t] \nonumber\\
	& \quad + \hbm \gamma \Phi[t-1] \bm a_n[t] - \gamma \hbm r_n[t-1] \Big)\Big  \rVert_2 \nonumber \\
	& = \Big \lVert \left (\bm I - \alpha_t\gamma \hbm \Phi[t-1]- \alpha_t  U_n[t] \hbm g[t] \hbm g^\top [t] \right)\bm a_n[t]  \nonumber \\
	& \quad+  \alpha_t  U_n[t] \tilde y_n[t]  \hbm g[t] +   \alpha_t \gamma\hbm r_n[t-1] \Big  \rVert_2.
\end{align}
Applying triangular inequality and  by assumption A\ref{as:mineig}, we have
\begin{subequations}
\begin{align}
	&\left \lVert \bm a_n[t+1] \right  \rVert_2\nonumber\\	
	& \leq  \lambda_{\mathrm{max}} \left (\bm I - \alpha_t\gamma \hbm \Phi[t-1] - \alpha_t U_n[t] \hbm g[t] \hbm g^\top [t] \right)\left \lVert\bm a_n[t] \right \rVert_2\nonumber \\
	& \quad + \alpha_t\left \lVert U_n[t] \tilde y_n[t]  \hbm g[t]\right  \rVert_2  +   \alpha_t \gamma\left \lVert\hbm r_n[t-1] \right  \rVert_2\\
	& = 1- \alpha_t\gamma\lambda_{\mathrm{min}} \left (\hbm \Phi[t-1] + \alpha_t U_n[t] \hbm g[t] \hbm g^\top [t] \right)\left \lVert\bm a_n[t] \right \rVert_2 \nonumber \\
	& \quad +  \alpha_t\left \lVert U_n[t] \tilde y_n[t]  \hbm g[t]\right  \rVert_2 +   \alpha_t\gamma \left \lVert\hbm r_n[t-1] \right  \rVert_2\\
	& \leq 1- \alpha_t\gamma\lambda_{\mathrm{min}} \left (\hbm \Phi[t-1] \right)\left \lVert\bm a_n[t] \right \rVert_2 +  \alpha_t\left \lVert U_n[t] \tilde y_n[t]  \hbm g[t]\right  \rVert_2 \nonumber \\
	& \quad+   \alpha_t \gamma \left \lVert\hbm r_n[t-1] \right  \rVert_2\\
	& \leq \left( 1- \alpha_t \gamma\strongcvxparf  \right)\left \lVert\bm a_n[t] \right \rVert_2 +  \alpha_t\left \lVert U_n[t] \tilde y_n[t]  \hbm g[t]\right  \rVert_2 \nonumber \\
	& \quad+   \alpha_t \gamma\left \lVert\hbm r_n[t-1] \right  \rVert_2. \label{eq:boundonageneral}
\end{align}
\end{subequations}
Substituting the bound on $\lVert\hbm r_n[t-1]\rVert_2$ from \eqref{eq:boundonr}	 into the above expression, we have
\begin{align}
	&\left \lVert \bm a_n[t+1] \right  \rVert_2 \nonumber \\
	&\leq \left( 1- \alpha_t \strongcvxparf \gamma \right)\left \lVert\bm a_n[t] \right \rVert_2+ \alpha_t\Big ( U_n[t]   \sqrt{PN}\energybound
	+\frac{\sqrt{PN} \energybound}{1-\gamma }\Big).
\end{align}
Setting $\alpha_t=\alpha$ and for $0<\alpha\le 1/L $, it can be proven by recursively substituting into \eqref{eq:boundonageneral} (similar steps to those in the proof of \cite[Theorem 5]{zaman2019online}), that
\begin{align}\label{eq:bound_a}
	\left \lVert \bm a_n[t+1] \right  \rVert_2 & \leq \frac{1}{\strongcvxparf\gamma }\left ( \frac{\nu}{1+\nu}  \sqrt{PN}\energybound+\frac{\sqrt{PN} \energybound}{1-\gamma }\right) \quad \forall t.
\end{align}
Substituting the above bound into \eqref{eq:gradboundJSTIRSO} completes the proof.
\section{Proof of Lemma 3}
\cmt{Error analysis for JSTIRSO}The error in the gradient for JSTIRSO is given by \eqref{eq:erroringrad} and can be rewritten as:
	\begin{align}
		\bm e^{(n)}[t] 
		&= U_n[t] (\hbm g[t]\hbm g^\top [t]\! -\!\bm g[t]\bm g^\top [t]) \bm a_n[t]  \nonumber \\
		& \quad + U_n[t] \tilde y_n[t]( \bm g[t] -\hbm g[t]) + \gamma (\bm r_n[t-1]-\hbm r_n[t-1]) \nonumber \\
		& \quad + \gamma( \hbm \Phi[t\!-\!1]\!-\! \bm \Phi[t\!-\!1]) \bm a_n[t] .
	\end{align}
Next, we take the norm on both sides of the above equation
	\begin{align}
		&\left \lVert \bm e^{(n)}[t]\right \rVert_2 
		\leq 
		\Big \lVert U_n[t] (\hbm g[t]\hbm g^\top [t] -\bm g[t]\bm g^\top [t]) \bm a_n[t]\Big \rVert_2
		\nonumber \\
		&  \quad+ \left \lVert\gamma( \hbm \Phi[t-1]-\bm \Phi[t-1]) \bm a_n[t] \right \rVert_2  \nonumber \\
		& \quad  + \left \lVert \gamma (\bm r_n[t-1]-\hbm r_n[t-1]) \right \rVert_2  + \left \lVert  U_n[t]\tilde y_n[t]( \bm g[t] -\hbm g[t]) \right \rVert_2 \nonumber\\
		& \leq \gamma \lambda_{\mathrm{max}} \left ( \hbm \Phi[t \!-\!1]\!-\!  \bm \Phi[t\!-\!1]\right) \left \lVert \bm a_n[t] \right \rVert_2 \!+\!   U_n[t]  \lambda_{\mathrm{max}} \big (  \hbm g[t]\hbm g^\top [t] \nonumber \\
		& \quad -\bm g[t]\bm g^\top [t] \big) \left \lVert \bm a_n[t]\right \rVert_2   + \gamma B_{\bm r} + U_n[t] \left \lvert\tilde y_n[t]\right \rvert \left \lVert   \bm g[t] -  \hbm g[t] \right \rVert_2, \label{eq:bound_e_1}
	\end{align}
where the first inequality holds because of the triangular inequality and the second inequality holds because of the Cauchy-Schwarz inequality. 

Besides, combining A\ref{as:boundedprocess} and \eqref{eq:bg} it can be proven that 
\begin{equation}
	\label{eq:bgg}
	\lambda_{\mathrm{max}} \left (  \hbm g[t]\hbm g^\top [t] \!-\!\bm g[t]\bm g^\top [t] \right) \leq 2 \sqrt{PN\energybound} B_{\bm g}+B_{\bm g}^2.
\end{equation}
By substituting  \eqref{eqs:boundederrors} and \eqref{eq:bgg}  into \eqref{eq:bound_e_1}, we obtain
\begin{subequations}
	\begin{align}
		&\left \lVert \bm e^{(n)}[t]\right \rVert_2 \nonumber \\
		& \leq \gamma B_{\bm\Phi} \left \lVert \bm a_n[t] \right \rVert_2 + U_n[t]\left(2 \sqrt{PN\energybound} B_{\bm g}+B_{\bm g}^2\right) \left \lVert \bm a_n[t] \right \rVert_2 \nonumber \\
		& \quad + \gamma B_{\bm r}+ U_n[t] B_{\bm g} \sqrt{\energybound}\label{eq:errorgradbound0}\\
		& \leq \gamma B_{\bm\Phi} \left \lVert \bm a_n[t] \right \rVert_2 + \left(\frac{\nu}{1+\nu}\right)\left(2 \sqrt{PN\energybound} B_{\bm g}+B_{\bm g}^2\right) \left \lVert \bm a_n[t] \right \rVert_2 \nonumber \\
		& \quad + \gamma B_{\bm r}+ \left(\frac{\nu}{1+\nu}\right) B_{\bm g} \sqrt{\energybound}  \label{eq:errorgradbound11} \\
		&= \left (\gamma B_{\bm\Phi}  + \left(\frac{\nu}{1+\nu}\right)\left(2 \sqrt{PN\energybound} B_{\bm g}+B_{\bm g}^2\right) \right ) \left \lVert \bm a_n[t] \right \rVert_2  \nonumber \\
		& \quad + \gamma B_{\bm r}+ \left(\frac{\nu}{1+\nu}\right) B_{\bm g} \sqrt{\energybound},  \label{eq:errorgradbound1}
	\end{align}
\end{subequations}
where the final result comes from substituting an upper bound on $U_n[t]$ and rearranging terms.
We can use here the same bound on $ \lVert \bm a_n[t] \rVert_2$ that was derived in the proof of Lemma \ref{lemma:boundedgrad} [cf. \eqref{eq:bound_a}]:
\begin{align}
	\left \lVert \bm a_n[t+1] \right  \rVert_2 & \leq \frac{\sqrt{PN} \energybound}{\strongcvxparf }\left (\frac{\nu}{1+\nu}+\frac{1}{1-\gamma }\right) \quad \forall ~t;
\end{align}
substituting the above bound into
\eqref{eq:errorgradbound1} completes the proof.
\newpage
\renewcommand{\thefigure}{6}
	\begin{figure*}
		\centering
		\includegraphics[width=\textwidth,height=24cm]{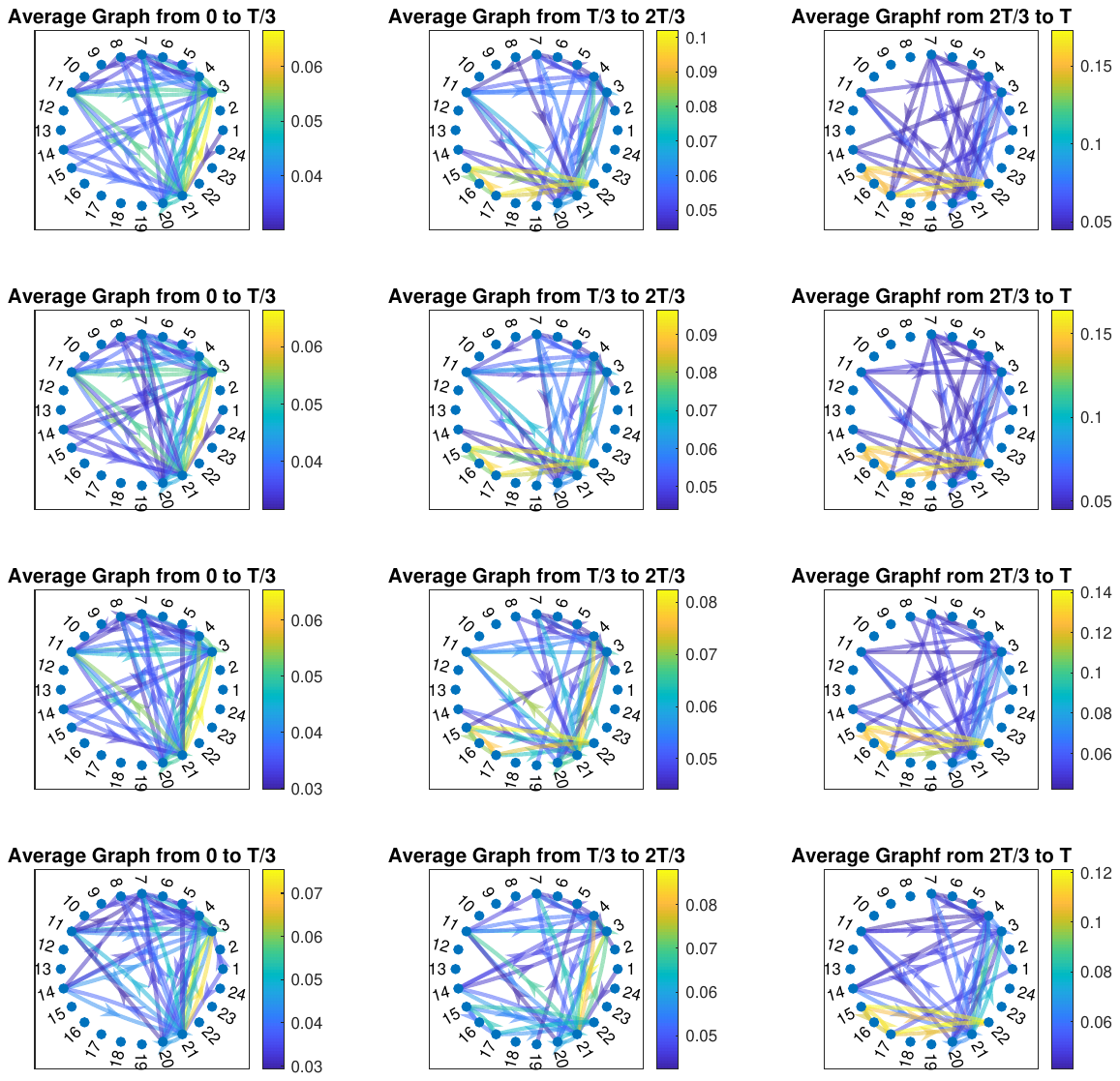}
		\caption{Average estimated graphs via JSTIRSO. The rows of the sub-figures correspond to missing probabilities 0, 0.05, 0.15, and 0.2. respectively. All the hyperparameters of the algorithm are computed via grid search.  Simulation parameters: $N=24, P=6, T=360, \sigma_u=0.01, \sigma_\epsilon=0.01, \gamma=0.9$.}
		\label{fig:real-data-avgGraphs}
	\end{figure*}
\end{document}